\documentclass[table,xcdraw]{article}

\usepackage[utf8]{inputenc}
\usepackage[T1]{fontenc}
\usepackage[english]{babel}
\usepackage{helvet} % Helvetica font

\usepackage[a4paper]{geometry}
\usepackage{amsmath, amssymb, amsfonts, mathtools, mathrsfs}
\usepackage{blkarray}  % Block arrays
\usepackage{numprint}

\newtheorem{theorem}{Theorem}[section]
\newtheorem{lemma}{Lemma}[section]
\newtheorem{proposition}{Proposition}[section]

\newenvironment{proof}{\textbf{Proof.}}{$\Box$}

\usepackage{graphicx}
\usepackage{pgfplots}
\usepackage{tikz}
\usepackage{tikz-network}
\usepackage{caption}
\usepackage{subcaption}
\usepackage{lscape}
\usepackage{float}
\usepackage{framed}
\usepackage{fancybox}
\usepackage[all]{xy}

\usepackage{longtable}
\usepackage{multirow}
\usepackage{colortbl}
\usepackage{algorithm}
\usepackage{algpseudocode}

\usepackage{hyperref}
\usepackage{color}
\usepackage[some]{background}

\usepackage{multicol}
\usepackage[strict]{changepage}
\usepackage{textcomp}
\usepackage{natbib}
\usepackage{eurosym}
\usepackage{appendix}
\usepackage{lineno}

\newcommand{\uvec}[1]{\boldsymbol{\textbf{#1}}}

\def\SS{{\mathbb{S}}}

\title{Two-stage indirect determinantal sampling designs}
\author{V. Loonis\footnote{INSEE, Paris, France. Email: vincent.loonis@insee.fr}.}
\date{}
\begin{document}
\maketitle

\begin{abstract}
A key feature of determinantal sampling designs is their capacity to provide known and parametrisable inclusion probabilities at any order. This paper aims to demonstrate how to effectively leverage this characteristic, highlighting its implications by addressing a practical challenge that arises when managing a network of face-to-face surveyors. This challenge is formulated as an optimization problem within the framework of two-stage indirect sampling, utilizing the Generalized Weight Share Method (GWSM). A general closed-form expression for the optimal weight matrix defined by the GWSM is derived, and based on a reasonable hypothesis, a formula for the optimal inclusion probabilities used in the second stage is provided. The implementation of the global optimization process is illustrated with real data, assuming that the intermediate and the second stage sampling designs are determinantal. Additionally, given these designs, closed-form expressions for the target first-order and joint inclusion probabilities are presented, thus paving the way for an alternative application of the Horvitz-Thompson estimator for evaluating any total within the target population. In short, determinantal sampling designs prove to be a versatile and useful tool for addressing practical challenges involving high-order inclusion probabilities.
\end{abstract}

\section*{Introduction}

Determinantal point processes are point processes widely used in various fields: random matrices, mathematical physics and machine learning. Their application in survey sampling leads to the class of determinantal sampling designs introduced by \citet{loonis2019determinantal}. These designs represent a broad family of sampling methods parameterised by the set of Hermitian contracting matrices \citep{macchi1975coincidence, soshnikov2000determinantal}. Let $U$ be a population of size $N$, and $\mathbf{\mathbf{K}}$, of size $(N\times N)$, be one of these matrices. Thus $DSD(\mathbf{\mathbf{K}})$ will denote the corresponding determinantal sampling design on $U$. The matrix $\mathbf{\mathbf{K}}$ is referred to as the kernel of $DSD(\mathbf{\mathbf{K}})$. A random sample $\SS$ from $U$ whose law is determinantal is thus written $\SS \sim  DSD(\mathbf{\mathbf{K}})$. Theoretically, any quantity describing $DSD(\mathbf{\mathbf{K}})$ can be expressed as a function of $\mathbf{\mathbf{K}}$. This might include, for example, the inclusion probabilities at any order or the variance of the Horvitz-Thompson estimator \citep{horvitz1952generalization}. Consequently, by applying suitable optimisation algorithms over the set of Hermitian contracting matrices, it is possible to derive the optimal determinantal sampling design based on a given optimisation criterion. For instance, \citet{loonis2019determinantal} and \citet{loonis2023} employed such a strategy to develop balanced determinantal sampling designs. Such designs, whether determinantal or not, aim to provide near-perfect estimates of a known parameter, that is a function of some auxiliary information. The underlying hypothesis assumes that a balanced sampling design, providing accurate estimates for a known parameter, should also be effective for estimating an unknown parameter that is a function of a variable correlated with the auxiliary information \citep{deville2004efficient}. The potential applications of determinantal sampling designs extend far beyond the given example. In this regard, indirect sampling \citep{deville2006indirect} offers a particularly promising avenue for harnessing their advantageous properties.
\paragraph{}
The general framework for this indirect sampling involves two populations, $U^A$ and $U^B$, with respective sizes $N^A$ and $N^B$. $U^A$ serves as an intermediate population, with its units indexed as $i$ (and $j$). $U^B$ is the target population, with its units designated $k$ (and $l$). A correspondence exists between the units of both populations, described by a $(N^A \times N^B)$ link matrix $\mathbf{L}^{\!\!AB}=[\ell^{AB}_{ik}]$; in the case where units $i$ and $k$ are linked then $\ell_{ik}^{AB} = 1$ , and $\ell_{ik}^{AB} = 0$ in the contrary case. A random sample $\SS^A$ is drawn from $U^A$ according to a given sampling design $\mathcal{P}^A$, whose first order inclusion probabilities $\pi^{A}_i=\Pr(\{i\}\subseteq \SS^A)$ are strictly positive and often prescribed, such that $\pi^{A}_i=\Pi^{A}_i$, where $\mathbf{\Pi}^A={(\Pi^A_1,\cdots,\Pi^A_{N^A})}^\intercal$ is a given vector of probabilities. The units in $U^B$ that are linked to those in $\SS^A$ are included in the target random sample $\SS^{1B}$, whose generally unknown probability law is denoted by $\mathcal{P}^{1B}$. The selection process of $\SS^{1B}$ is generally referred to as indirect sampling or to as one-stage indirect sampling in this article (see Figure \ref{Fig1}).
\begin{figure}[h!]
\begin{subfigure}[t]{0.47\textwidth}
\centering
\includegraphics[scale=1]{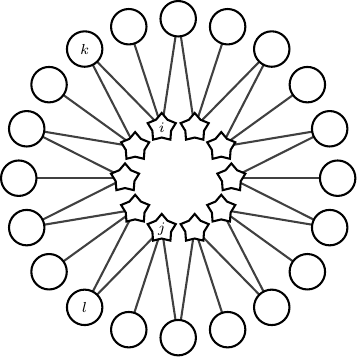}\vfill
    \caption{\raggedright Two related populations: $U^A$ (star), $U^B$ (circle)\label{FigIndia}, $\ell^{AB}_{ik}=1$ (edge)}
    \end{subfigure}\label{Fig1a}
    \begin{subfigure}[t]{0.47\textwidth}
\centering
\includegraphics[scale=1]{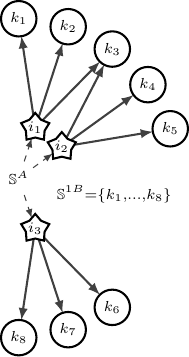}\vfill
    \caption{\raggedright Drawing $\mathbb{S}_A$ from $U_A$, deducing $\SS^{1B}$ from $U^B$}
    \end{subfigure}
    \caption{One-stage indirect sampling}\label{Fig1}
\end{figure}
\paragraph{}
Computing the first order target inclusion probabilities, $\pi^{1B}_k=\Pr(\{k\}\subseteq \SS^{1B}),$ $k=1,\cdots,N^B,$ happens to be a difficult task. To evaluate them, one has to calculate: 
\begin{equation}
    \pi_k^{1B} = \Pr\left( \exists i \in U_k^A \mid i \in \mathbb{S}^A \right) = 1 - \Pr\left( U_k^A \subseteq {\mathbb{S}^A}^{\,c} \right),
\end{equation} where $U^A_k$ stands for the set of units of $U^A$ that are linked with $k$: $U^A_k=\{i \in U^A | \ell^{AB}_{ik}=1\}$, and ${\SS^A}^c$ for the complementary of $\SS^A$ in $U^A$. That is to say one has to know the inclusion probabilities of potentially high orders of ${\mathcal{P^A}}^c$, being the probability law of ${\SS^A}^c$. Apart from some intermediate sampling designs, such as  stratified simple random sampling \citep{lavallee2013indirect,Bodet2022}, these quantities are unreachable, and the $\pi^{1B}_k$ remain unknown. This result is even truer for the joint inclusion probabilities $\pi^{1B}_{kl}=\Pr(\{k,l\}\subseteq \SS^{1B})$, or the target inclusion probabilities of a higher order. Consequently, the utilisation of the Horvitz-Thompson estimator, $\smash{\hat{t}^{HT}_{\mathbf{y}^B}=\sum_{k\in \SS_1^B}y^B_k/\pi_{k}^{1B}}$, to estimate the total value $t_{\mathbf{y}^B}$ of a variable of interest $\mathbf{y}^B$ on the target population is often precluded. 

\paragraph{}
To overcome the general lack of information regarding the target inclusion probabilities for estimating $t_{\mathbf{y}^B}$, \citet{deville2006indirect} propose the utilisation of an alternative unbiaised parametrised class of linear estimators, whose weights are random variables. This method is referred to as the General Weight Share Method (GWSM). The multidimensional parameter, chosen by the survey statistician, consists of a $(N^A\times N^B)$ weight matrix $\mathbf{\Theta}^{\!AB}=[\theta^{AB}_{ik}]$, subject to certain constraints that guarantee the unbiasedness of the estimators.
\paragraph{}
Indirect sampling and the GWSM are particularly valuable where there is no sampling frame available to facilitate direct random unit selection from $U^B$, while one does exist for a related population $U^A$. This situation arises, for example, when sampling homeless individuals ($U^B$) via the various social services ($U^A$) provided to them \citep{vitiis2014implementing}; tourists ($U^B$) via their visits to specific locations ($U^A$) \citep{deville2006extension}; or postmen ($U^B$) via the addresses ($U^A$) they deliver to \citep{medous2025optimal}. In such cases, it is not necessary to know the entire matrix $\mathbf{L}^{\!\!AB}$ to implement a \textit{minimal} version of the GWSM, which may well result in a high variance of the GWSM estimators. Beyond these practical applications, indirect sampling has theoretical implications for the weighting of panel surveys \citep{ernst2009weighting}, multi-frame sampling \citep{deville2006indirect}, and capture-recapture processes \citep{lavallee2012capture}. In this article, we will further explore the coordination of spatial sampling, drawing on the examples provided by \citet{Hall2008} and \citet{Cheval2022}. In this latter case, a sampling frame for $U^B$ might exist, but various constraints make indirect sampling more suitable. The existence of a sampling frame for $U^B$ is therefore significant. It implies that $\mathbf{L}^{\!\!AB}$ is fully known, making the search for optimal GWSM estimators easier.
 
\paragraph{}
To help select the appropriate parameter matrix, when $\mathbf{L}^{\!\!AB}$ is fully known, \citet{deville2006indirect} identify the optimal matrix $\mathbf{\Theta}^{\!AB}_{opt}(\mathbf{y}^B)$ that minimises the variance of the GWSM estimators for measuring $t_{\mathbf{y}^B}$. In the case where $\mathbf{\Theta}^{\!AB}_{opt}(\mathbf{y}^B)$ does not depend on $\mathbf{y}^B$, optimality is defined as being \textit{strong}. Except for certain specific intermediate sampling designs, such as Poisson sampling or Simple Random Sampling (SRS), and particular types of link matrices referred to as \textit{many-to-one} \citep{medous2025optimal}, achieving this condition is in practice very difficult.  As a consequence,an optimal matrix that is independent of $\mathbf{y}^B$ may not exist. This being so, the authors consider a less demanding optimality which they refer to as \textit{weak} optimality. It consists in looking for the GWSM estimator whose variance of the random weights is minimal. To achieve this goal, the authors show that \textit{weak} optimality is equivalent to considering the matrices $\mathbf{\Theta}^{\!AB}_{opt}(\mathbf{y}^B_q)$ for the very specific and finite set $\{\mathbf{y}^B_1,\cdots,\mathbf{y}^B_{N^B}\}$, such that $\mathbf{y}^B_q=e_q$, the $q^{th}$ canonical vector of $\smash{{\mathbb{R}^{N}}^{\scriptscriptstyle B}}$. They show there is indeed a unique matrix that minimises the variance of the GWSM estimators for measuring each $t_{\mathbf{y}_q^B}$, $q=1,\cdots,N^B$.

\paragraph{}For any other given variable $\mathbf{y}^B$,  \citet{deville2006indirect} argue that the optimal matrix $\mathbf{\Theta}^{\!AB}_{opt}(\mathbf{y}^B)$ has limited usefulness. Calculating this matrix implies knowledge of $y^B_k$ for all the units of the population $U^B$. In practice this assumption is unrealistic. If this was the case a sample for estimating $t_{\mathbf{y}^B}$ would not need to be drawn. In this article, given that a sampling frame is available for $U^B$ and that $\mathbf{L}^{\!\!AB}$ is fully known, we can nuance their reasoning by further considering a set of auxiliary variables $\mathbf{X}^B$. A matrix $\smash{\mathbf{\Theta}^{\!AB}_{opt}(\mathbf{X}^B)}$ will then be optimal if it achieves, to the extent possible, a balanced indirect sampling with respect to $\mathbf{X}^B$. 
\paragraph{}
A final practical challenge in the implementation of $\mathbf{\Theta}^{\!AB}_{opt}(\mathbf{X}^B)$ lies in the need to know the values of the joint inclusion probabilities $\pi^{A}_{ij} = \text{pr}(\{i,j\} \subseteq \mathbb{S}^A)$ for all the pairs of units $(i,j)$ within $U^A \times U^A$. For most typical and somewhat complex $\mathcal{P}^A$ sampling designs, these probabilities remain unknown. As a result, $\mathbf{\Theta}^{\!AB}_{opt}(\mathbf{X}^B)$ often cannot be computed in practice. 
\paragraph{}
In this article, we demonstrate that the general framework of indirect sampling and the GSWM described above can benefit in many respects from an intermediate determinantal sampling design: $\mathcal{P}^A = DSD(\mathbf{\mathbf{K}}^A)$. First, the target simple and joint inclusion probabilities become accessible, whatever the kernel $\mathbf{\mathbf{K}}^A$. Moreover, as these probabilities are parameterised by $\mathbf{\mathbf{K}}^A$, one can adjust $\mathbf{\mathbf{K}}^A$ to achieve, as closely as possible, specific desired features. These include, for example, prescribed target first-order probabilities $\pi^{1B}_k = \Pi^{1B}_k$, where $\mathbf{\Pi^{1B}} = (\Pi^{1B}_1, \cdots, \Pi^{1B}_{N^B})^\intercal$ is a given vector of probabilities.
Secondly, $\mathbf{\Theta}^{\!AB}_{opt}(\mathbf{X}^B)$ can be computed in practice for any kernel $\mathbf{\mathbf{K}}^A$, as the intermediate joint probabilities $\pi^{A}_{ij}$ are known. Here again, since these probabilities are parameterised by $\mathbf{\mathbf{K}}^A$, this kernel can be fine-tuned to minimize even further the optimal variance of the GSWM estimators over the set of Hermitian contracting matrices.
\paragraph{}
Finally, these advantages of intermediate determinantal sampling designs remain relevant in the more general context of the two-stage indirect sampling. This approach, which consists in selecting a second sample from $\SS^{1B}$, is often employed to better manage the target sample size that is typically random in the one-stage case \citep{deville2006indirect}.

\paragraph{} 
To achieve our primary goal of showcasing the versatility of determinantal sampling designs, this article is structured into three sections. The first section revisits key results on indirect sampling and the GSWM, focusing on the two-stage case that encompasses that of the one-stage. While building on existing literature \citep{deville2006indirect,lavallee2007indirect}, we introduce several new contributions, including a closed-form expression for the two-stage variance of the GSWM (Proposition \ref{variance}), which directly leads to a closed-form formula for the two-stage optimal GSWM weight matrix (Proposition \ref{PropOpt}). Additionally, we provide a closed-form expression for the second-stage optimal inclusion probabilities, subject to an additional - yet reasonable - assumption (Proposition \ref{PiOpt}).  Building on \citet{loonis2019determinantal} and \citet{loonis2023}, the second section introduces determinantal sampling designs and the advantages they offer when applied to intermediate- or second-stage indirect sampling. Two main innovations are presented in this section: closed-form formulae for target first-order and joint inclusion probabilities (Proposition \ref{Probacible}) and improved methods for handling optimisation problems. The third section consolidates our findings and assesses their effectiveness in solving one of the particular problems that arise when managing a network of face-to-face surveyors. The core challenge can be reduced to drawing two approximately equal-sized samples from two distinct populations, with each sample balanced with respect to a specific set of auxiliary variables. Additionally, each unit selected in one sample should closely match a unit from the other sample based on a defined distance criterion (see Figure \ref{Fig2}). Our solution relies on the simultaneous use of indirect and determinantal sampling designs. 

\begin{figure}[h!]
\centering
\begin{subfigure}{0.47\textwidth}
\vtop{
\centering
\includegraphics[scale=0.4,trim=6pt 5pt 6pt 5pt, clip, width=0.8\textwidth]{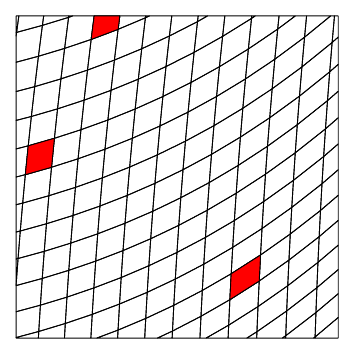}
\caption{\raggedright Drawing a sample $\mathbb{S}^A$ from $U_A$, \textit{balanced} with respect to a set $\mathbf{X}^A$ of auxiliary variables.}
}
\end{subfigure}
\hfill
\begin{subfigure}{0.47\textwidth}
\vtop{
\centering
\includegraphics[scale=0.4,trim=6pt 5pt 6pt 5pt, clip, width=0.8\textwidth]{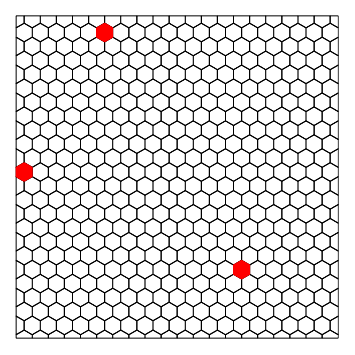}
\caption{\raggedright Drawing a close equal-sized sample $\mathbb{S}^{2B}$ from $U^B$, \textit{balanced} with respect to a set $\mathbf{X}^B$ of auxiliary variables.}
}
\end{subfigure}
\caption{Coordinating samples for face-to-face surveys}\label{Fig2}
\end{figure}
\section{Two-stage indirect sampling and the GWSM}
\subsection{Notations}

In this section and throughout the article, \( \mathbf{I}_N \) denotes the \( (N \times N) \) identity matrix, and \(\mathbf{J}_N \) denotes the \( (N \times N) \) all-ones matrix. The vector \( \uvec{1}_N \in \mathbb{R}^N \) consists exclusively of ones. For any matrix \( \mathbf{M} \), the vector \( \mathbf{v}_\mathbf{M} \) is defined as the column vector obtained by stacking the columns of \( \mathbf{M} \) on top of each other. The notation \( \mathbf{M}_{\mid s} \) refers to the submatrix of \( \mathbf{M} \) formed by selecting the rows and columns indexed by the set \( s \subseteq \{1,\dots,N\} \). 
For any vector \( \mathbf{x} = (x_1, \cdots , x_N)^\intercal \in \mathbb{R}^N \), the matrix \( \mathbf{D}_\mathbf{x} \) denotes the \( (N \times N) \) diagonal matrix whose diagonal entries are \( x_1, \cdots, x_N \). The symbols \( \odot \) and \( \otimes \) respectively denote the Hadamard (element-wise) product and the Kronecker product between matrices. A variable of interest is defined as a variable whose values are known, in practice, only for the selected sample units. In contrast, an auxiliary variable is observed for the entire population and can be exploited to improve sampling procedures to estimate a function of the variable of interest. These two types of variables are denoted by $\mathbf{y}$ and $\mathbf{x}$, respectively, and are represented as vectors in $\mathbb{R}^N$.

\subsection{General settings}
Based on the presentation and the notations introduced earlier, the first stage  of indirect sampling involves drawing a random sample $\SS^A$ from $U^A$ and identifying the target sample $\mathbb{S}^{1B}$, which includes all units from $U^B$ that are linked to at least one unit in $\SS^A$: 
\begin{equation}\label{Equation_Indirect}
\SS^{1B}=\bigcup_{i\in \SS^A} U^B_i,
\end{equation}
where, symmetrically to $U^A_k$, $U^B_i=\{k \in U^B | \ell^{AB}_{ik}=1\}$, $i=1,\cdots,N^A$.
\paragraph{}
The second stage consists in selecting, for each unit $i$ in $\mathbb{S}^A$, a subsample of units from $U^B_i$. Let $\mathbb{S}^B_i$ be a random sample from $U^B_i$, whose probability law is the sampling design $\mathcal{P}^B_i$, $i=1,\cdots,N^A$, that verifies the following classical assumption for the two-stage sampling designs: 

\begin{equation}\label{H1}
H1\text{ : }\forall i=1,\cdots,N^A\text{, }\SS_i^B\text{ is independent from }\SS^A.
\end{equation}

Unless explicitly stated otherwise, it is not required that $\SS_i^B$ be independent from $\SS_j^B$ for all $i \neq j$. The global set $\smash{\mathcal{P}^{AB}=\{\mathcal{P}_i^B\}_{i=1}^{N^A}}$ defines a sampling design on $\smash{U^A \times U^B}$. Its first-order inclusion probabilities are given by $\pi^{AB}_{(i,k)}=\Pr({k}\subseteq \mathbb{S}^B_i)$. For any unit $k$ not belonging to $U^B_i$, it naturally follows that $\pi^{AB}_{(i,k)}=0$. Consistent with classical unbiasedness constraints, it is assumed that $\pi^{AB}_{(i,k)}>0$ for the remaining units. These probabilities are assembled in the $(N^A \times N^B)$ matrix $\mathbf{\Pi}^{1AB}$. The joint probabilities of $\mathcal{P}^{AB}$ are represented by the quantities $\pi^{AB}_{(i,k)(j,l)}=\Pr({k}\subseteq \mathbb{S}^B_i,{l}\subseteq \mathbb{S}^B_j)$, which are stored in the $(N^A N^B \times N^A N^B)$ matrix $\mathbf{\Pi}^{2AB}$. Section \ref{implement} provides details on how to implement this matrix, with particular emphasis on its specific indexation.. 
\paragraph{}
With these tools, the two-stage indirect sampling, described in Figure \ref{Fig3}, writes:
$$
\SS^{2B}=\bigcup_{i\in \SS^A} \SS^B_i.$$
%To conclude this part by noting that although one of the goals of the two-stage indirect sampling design is to help manage the final target sample size, the latter remains random. Indeed, a unit $k$ in $U^B$ can be selected several times, once for each unit in $\SS^A \cap U_k^A$. Fixing the final size requires considering a global sampling design $\mathcal{P}^{AB}$ such that $\pi^{AB}_{ik,jk}=0$ for all $k$ and for all $i$ different from $j$. Developing such a process appears to be a challenging - yet promising - task. It will not be addressed in this article.%

\subsection{Estimating a total }\label{Sec_Total}
For given matrices $\mathbf{L}^{\!\!AB}$ and $\mathbf{\Theta}^{\!AB}$, the following class of linear estimators that are used at the second stage to estimate $t_{\mathbf{y}^B}$ is derived from the GWSM: 

\begin{align}\label{Estimator}
\hat{t}_{\mathbf{y}^B} &= \sum_{k \in U^B} w_k y^B_k, \quad \text{where} \\
w_k &= \sum_{i \in U^A} \frac{\ell_{ik} \theta^{AB}_{ik} \uvec{1}_{\{i \in S^A\}} \uvec{1}_{\{k \in S_i^B\}}}{\pi^A_i \pi^{AB}_{(i,k)}}\nonumber \\ 
&= \sum_{i \in U^A_k} \frac{\theta^{AB}_{ik} \uvec{1}_{\{i \in S^A\}} \uvec{1}_{\{k \in S_i^B\}}}{\pi^A_i \pi^{AB}_{(i,k)}} \nonumber
\end{align}

\paragraph{} $\hat{t}_{\mathbf{y}^B}$ will be unbiased if, for all units in $U^B$,  $\mathbb{E}(w_k)=\mathbb{E}_{\SS^A}(\mathbb{E}(w_k|\SS^A))=1$. This can be expressed as: $$\smash{\sum_{i \in U^A}\ell_{ik}\theta^{AB}_{ik}=\allowbreak\sum_{i \in U^A_k}\theta^{AB}_{ik}=1}.$$  This condition is significant. It implies that implementing $\smash{\hat{t}_{\mathbf{y}^B}}$ requires knowing, for a unit $k$ in $\SS^{2B}$,  the list of units belonging to $U^A_k$, whether they belong to $\mathbb{S}^A$ or not. With this list, $\mathbf{\Theta}^{\!AB}$ can be chosen so that $\smash{\sum_{i \in U^A}\ell_{ik}\theta^{AB}_{ik}=\sum_{i \in U^A_k}\theta^{AB}_{ik}=1}$. For each selected $k$, the minimal supplementary information needed is the list of units of $U^A_k$ that do not belong to $\SS^A$ (see Figure \ref{Fig2b}). In practice, this supplementary information can be obtained via administrative data or by adding some specific questions to the questionnaire.
\paragraph{}
In the sequel, the unbiasedness constraints will be referred as to:
\begin{equation}\label{H2}H2a\text{ : }(\mathbf{L}^{\!\!AB} \odot \mathbf{\Theta}^{\!AB})^\intercal\uvec{1}_{N^A}=\uvec{1}_{N^B}\Leftrightarrow \forall k \in U^B, \sum_{i\in U^A}\ell_{ik} \theta^{AB}_{ik}=\sum_{i\in U^A_k} \theta^{AB}_{ik}=1.\end{equation} 

Unlike \citet{deville2006indirect} the entries of $\mathbf{\Theta}^{\!AB}$ are not required to be positive. While this condition is indeed helpful to avoid obtaining a negative estimate of a sum of positive values, it is not necessary to ensure that the estimator is unbiased.

\begin{figure}[h!]
\begin{subfigure}[t]{0.47\textwidth}
\centering
\includegraphics[scale=1]{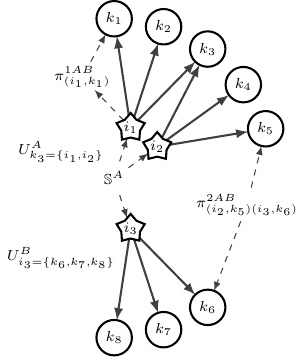}
    \caption{\raggedright Preparing the second-stage}\label{Fig2a}\vfill
    \end{subfigure}
    \begin{subfigure}[t]{0.55\textwidth}
\centering
\includegraphics[scale=1]{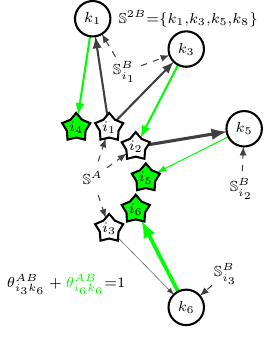}
    \caption{\raggedright Drawing $\mathbb{S}^{2B}$, deducing the supplementary information needed to satisfy (\ref{H2}) (green): $k_6$ was drawn from $i_3$, but could have been drawn from $i_6$,  according to Figure \ref{Fig1}.}\vfill\label{Fig2b}
    \end{subfigure}
    \caption{Second stage of a two-stage indirect sampling}\label{Fig3}
\end{figure}

\subsection[Computing the variance of t-hat]{Computing the variance of $\hat{t}_{\mathbf{y}^B}$}

The variance of $\hat{t}_{\mathbf{y}^B}$ can be expressed in a closed-form formula, with the proof provided in Appendix  (Section \ref{Proofvariance}).

\begin{proposition}[Variance of $\hat{t}_{\mathbf{y}^B}$]\label{variance}
The variance of the estimator $\hat{t}_{\mathbf{y}^B}$ is given by:
\[
\operatorname{var}(\hat{t}_{\mathbf{y}^B}) = 
\mathbf{v}_{\mathbf{\Theta}^{\!AB}}^\intercal
\mathbf{\mathcal{Q}}(\mathbf{L}^{\!\!AB}, \mathbf{y}^B {\mathbf{y}^B}^\intercal, \tilde{\mathbf{\Delta}}^{A}, \tilde{\mathbf{\Delta}}^{AB})
\mathbf{v}_{\mathbf{\Theta}^{\!AB}},
\]
where the symmetric matrix $\mathbf{\mathcal{Q}}(\cdot) \in \mathbb{R}^{N^A N^B \times N^A N^B}$ is defined as:
\[
\mathbf{\mathcal{Q}}(\mathbf{L}^{\!\!AB}, \mathbf{y}^B {\mathbf{y}^B}^\intercal, \tilde{\mathbf{\Delta}}^{A}, \tilde{\mathbf{\Delta}}^{AB}) = 
\mathbf{D}_{\mathbf{v}_{\mathbf{L}^{\!\!AB}}}
\left[
\left(\mathbf{y}^B {\mathbf{y}^B}^\intercal \otimes\mathbf{J}_{N^A} \right) \odot 
\left(
\mathbf{J}_{N^B} \otimes \tilde{\mathbf{\Delta}}^{A}
+ \tilde{\mathbf{\Delta}}^{AB}
+ \left(\mathbf{J}_{N^B} \otimes \tilde{\mathbf{\Delta}}^{A} \right) \odot \tilde{\mathbf{\Delta}}^{AB}
\right)
\right]
\mathbf{D}_{\mathbf{v}_{\mathbf{L}^{\!\!AB}}},
\]
with:
- $\tilde{\mathbf{\Delta}}^{A} \in \mathbb{R}^{N^A \times N^A}$ defined by
\[
\tilde{\Delta}^{A}_{ij} = 
\frac{\pi^{A}_{ij} - \pi^{A}_i \pi^{A}_j}{\pi^{A}_i \pi^{A}_j},
\]
- $\tilde{\mathbf{\Delta}}^{AB} \in \mathbb{R}^{N^A N^B \times N^A N^B}$ defined by
\[
\tilde{\Delta}^{AB}_{(i,k),(j,l)} = 
\begin{cases}
\displaystyle\frac{\pi^{AB}_{(i,k),(j,l)} - \pi^{AB}_{(i,k)}\pi^{AB}_{(j,l)}}{\pi^{AB}_{(i,k)}\pi^{AB}_{(j,l)}} & \text{if } \pi^{AB}_{(i,k)} > 0,\, \pi^{AB}_{(j,l)} > 0, \\
0 & \text{otherwise.}
\end{cases}
\]
\end{proposition}

In Proposition \ref{variance}, $\operatorname{var}(\hat{t}_{\mathbf{y}^B})$ consists of the sum of three terms corresponding to the  variability of $\hat{t}_{\mathbf{y}^B}$ due to $\mathcal{P}^A$, to $\smash{\mathcal{P}^{AB}}$ and to the interaction of both sampling designs. In the one-stage case, this corresponds to  all the $\mathcal{P}^{B}_i$ being censuses, one has $\smash{\tilde{\mathbf{\Delta}}^{AB}=0_{(N^AN^B\times N^AN^B)}}$ and:

\begin{equation}\smash{\operatorname{var}(\hat{t}_{\mathbf{y}^B})=\mathbf{v}_{\mathbf{\Theta}^{\!AB}}^\intercal \mathbf{D}_{\mathbf{v}_{\mathbf{L}^{\!\!AB}}}\left[ ((\mathbf{y}^B{\mathbf{y}^B}^\intercal)\otimes\mathbf{J}_{N^A})\odot (\mathbf{J}_{N^B}\otimes \tilde{\mathbf{\Delta}}^A)\right ]\mathbf{D}_{\mathbf{v}_{\mathbf{L}^{\!\!AB}}}\mathbf{v}_{\mathbf{\Theta}^{\!AB}}}. \end{equation} This form is equivalent to the classical one-stage formula given by \citet{deville2006indirect}: 
\begin{equation}\label{VarGSWMa}
\operatorname{var}(\hat{t}_{\mathbf{y}^B})=\left( \mathbf{y}^B \right)^{\intercal} {(\mathbf{L}^{\!\!AB}\odot\mathbf{\Theta}^{\!AB})}^\intercal \tilde{\mathbf{\Delta}}^A(\mathbf{L}^{\!\!AB}\odot\mathbf{\Theta}^{\!AB})\mathbf{y}^B. 
\end{equation}

\paragraph{}
The main advantage of Proposition \ref{variance} is that $\operatorname{var}(\hat{t}_{\mathbf{y}^B})$ naturally takes the form of a quadratic function of $\mathbf{v}_{\mathbf{\Theta}^{\!AB}}$. This advantage also holds when looking at a weighted sum of variances for a set of $Q$ auxiliary target variables $\mathbf{X}^B = (\mathbf{x}^B_1, \dots, \mathbf{x}^B_Q)$: $\sum_{q=1}^Q \alpha^B_q \text{var}(\hat{t}_{\mathbf{x}_q^B})$, where the weights ${\alpha^B_1, \dots, \alpha^B_Q}$ are positive values chosen by the statistician to reflect the relative importance of each variable. In this case, Proposition \ref{variance} directly provides:

\begin{equation} \smash{\sum_{q=1}^Q \alpha^B_q \text{var}(\hat{t}_{\mathbf{x}_q^B}) = \mathbf{v}_{\mathbf{\Theta}^{\!AB}}^\intercal \mathbf{\mathcal{Q}}(\mathbf{L}^{\!\!AB}, \sum_{q=1}^Q \alpha^B_q(\mathbf{x}^B_q {\mathbf{x}_q^B}^\intercal), \tilde{\mathbf{\Delta}}^A, \tilde{\mathbf{\Delta}}^{AB}) \mathbf{v}_{\mathbf{\Theta}^{\!AB}}}. \label{sumvar}\end{equation}
\paragraph{}
This formulation will simplify the derivation of a closed-form solution for the optimal weight matrix in the next section. 
\paragraph{}
However, a notable drawback is that implementing Proposition \ref{variance} requires handling matrices of size $(N^A N^B \times N^A N^B)$, which can be extremely large. In practice, these matrices have numerous rows, and the corresponding columns, filled entirely with zeros. As a result, by removing these lines and columns, they can be reduced to matrices whose size is determined by the number of non-zero elements in $\mathbf{L}^{\!\!AB}$ (see Section \ref{implement}). Therefore, a key objective, whenever possible, is to minimize the number of links between $U^A$ and $U^B$. In the application study, we will ultimately work with $N^A=250$  and $N^B=337$, while the final size of the matrices is \textit{only} $(631\times 631)$.

\subsection{Optimal weight matrix}

In this section $\mathcal{P}^A$ and $\mathcal{P}^{AB}$ are fixed, and consequently, so are $\tilde{\mathbf{\Delta}}^A$ and $\tilde{\mathbf{\Delta}}^{AB}$. The objective is to minimise (\ref{sumvar}), over the set of weight matrices $\mathbf{\Theta}^{\!AB}$ subject to (\ref{H2}). This Equality constraint (\ref{H2}) can also be expressed as a function of $\mathbf{v}_{\mathbf{\Theta}^{\!AB}}$: \begin{equation}\label{H1_nouveau}
    H2b \text{ : } \mathbf{E}^{AB}\mathbf{v}_{\mathbf{\Theta}^{\!AB}}=\uvec{1}_{N^B},
\end{equation} where $\mathbf{E}^{AB}$ is a $(N^B \times N^AN^B)$ bloc-diagonal matrix whose $k^{th}$ bloc is the transpose of the $k^{th}$ column of $\mathbf{L}^{\!\!AB}$ (see Section \ref{implement} for a practical implementation). Solving our minimisation problem becomes a relatively standard task of quadratic programming with equality constraints over $\mathbf{v}_{\mathbf{\Theta}^{\!AB}}$\citep{boyd2004convex}, leading to a closed-form formula for the optimal vector $\mathbf{v}_{\mathbf{\Theta}^{\!AB}_{opt}}$ and consequently for the optimal weight matrix $\mathbf{\Theta}^{\!AB}_{opt}$.

\begin{proposition}[Optimal weight matrix for a set $\mathbf{X}^B$ of auxiliary variables]\label{PropOpt}
Let $\mathbf{L}^{\!\!AB}$, $\mathcal{P}^A$, $\mathcal{P}^{AB}$, $\mathbf{X}^B = \{\mathbf{x}_1^B, \dots, \mathbf{x}_Q^B\}$ and weights $\alpha^B_1,\dots,\alpha^B_Q$ be given. The vector $\mathbf{v}_{\mathbf{\Theta}^{\!AB}_{opt}(\mathbf{X}^B)} \in \mathbb{R}^{N^AN^B}$ minimizing $\sum_{q=1}^Q \alpha^{B}_q \operatorname{var}(\hat{t}_{\mathbf{x}_q^B})$ under constraint~\eqref{H1_nouveau} is given by:
\[
\mathbf{v}_{\mathbf{\Theta}^{\!AB}_{opt}(\mathbf{X}^B)} = 
\left[
\begin{pmatrix}
\mathbf{\mathcal{Q}}(\mathbf{L}^{\!\!AB}, \sum_{q=1}^{Q} \alpha_q^B \mathbf{x}_q^B {\mathbf{x}_q^B}^\intercal, \tilde{\mathbf{\Delta}}^A,\tilde{\mathbf{\Delta}}^{AB}) & (\mathbf{E}^{AB})^\intercal \\
\mathbf{E}^{AB} & 0_{N^B \times N^B}
\end{pmatrix}^\dagger
\begin{pmatrix}
0_{N^AN^B} \\
\uvec{1}_{N^B}
\end{pmatrix}
\right]_{|\{1,\dots,N^AN^B\}}
+ \ker(\mathbf{\mathcal{Q}}(\cdot)).
\]
where $\dagger$ stands for the Moore-Penrose inverse.
\end{proposition}

A particular scenario arises when \(Q = N^B\) and \(\mathbf{x}_q^B = e_q\), the \(q^{th}\) canonical vector. In this case the matrix \(\mathbf{X}^B\) reduces to the identity matrix \(\displaystyle \mathbf{I}_{N^B}\), and \(\hat{t}_{\mathbf{x}_q^B} = w_q\). The optimization problem then consists of minimizing:
\[
\sum_{q=1}^Q \alpha_q^B \operatorname{var}(w_q).
\] 
Since each \(w_q\) depends on a distinct set of parameters \(\{\theta^{AB}_{1q}, \ldots, \theta^{AB}_{N^A q}\}\), the problem is equivalent to independently minimizing the variance of each random GWSM weight \(w_q\). This solution generalizes the (one-stage) \emph{weak} optimality studied by \citet{deville2006indirect}. Moreover, combining this \emph{weak} optimality with our approach for a set \(\mathbf{X}^B\) of auxiliary variables is made possible by considering the augmented set \((\mathbf{X}^B : \mathbf{I}_{N^B})\) in Proposition \ref{PropOpt}.

\subsection{Optimal inclusion probabilities at the second stage}
After optimizing \(\displaystyle \textstyle\sum_{q=1}^Q \alpha^B_q \operatorname{var}(\hat{t}_{\mathbf{x}_q^B})\) over \(\mathbf{\Theta}^{\!AB}\) this section investigates the optimization of this quantity over \(\tilde{\mathbf{\Delta}}^{AB}\) for a simplified problem, by incorporating the following assumptions in addition to (\ref{H1}): 
\begin{equation}\label{H3-4}
\left\{
\begin{array}{ll}
H3: & \forall (i,j) \in U^A \times U^A,\, i \neq j,\quad \mathbb{S}^B_i \text{ is independent from } \mathbb{S}^B_j, \\
H4: & \forall i \in U^A, \quad \mathbb{S}^B_i \text{ is of fixed size } 1.
\end{array}
\right.
\end{equation}

As a consequence, \(\tilde{\mathbf{\Delta}}^{AB}\) depends solely on \(\mathbf{\Pi}^{1AB}\), since we have:
\[
\begin{cases}
\pi^{AB}_{(i,k)(i,l)} = 0, & \text{for all } i \text{ and all } k \neq l, \\
\pi^{AB}_{(i,k)(j,l)} = \pi^{AB}_{(i,k)} \pi^{AB}_{(j,l)}, & \text{for all } i \neq j.
\end{cases}
\]

Therefore, minimizing \(\displaystyle \textstyle\sum_{q=1}^Q \alpha^B_q \operatorname{var}(\hat{t}_{\mathbf{x}_q^B})\) over \(\tilde{\mathbf{\Delta}}^{AB}\) is equivalent to minimizing it over \(\mathbf{\Pi}^{1AB}\). From this, we derive the optimal second-stage inclusion probabilities. The proof is provided in Appendix (Section \ref{ProofPiOpt}).

\begin{proposition}[Optimal second-stage inclusion probabilities]\label{PiOpt}
Let $\mathbf{L}^{\!\!AB}$ be a link matrix, $\mathcal{P}^A$ an intermediate sampling design for population $U^A$, and let $\mathbf{\Theta}^{\!AB}$ be a weight matrix satisfying constraint~\eqref{H2}. Let $\mathcal{P}^{AB}$ be a sampling design for $U^A\times U^B$ satisfying constraints~\eqref{H3-4}. Let $\mathbf{X}^B = \{\mathbf{x}_1^B, \dots, \mathbf{x}_Q^B\}$ be a set of $Q$ target auxiliary variables on $U^B$, with associated non-negative weights $\alpha_1^B, \dots, \alpha_Q^B$.

Then the optimal second-stage first-order inclusion probabilities $\pi^{AB,opt}_{(i,k)}$ minimizing the weighted total variance $
\sum_{q=1}^Q \alpha^B_q \operatorname{var}(\hat{t}_{\mathbf{x}_q^B}),
$
are given, for each $(i,k)$ with $k \in U_i^B$, by:
\[
\pi^{AB,opt}_{(i,k)} =
\frac{
|\theta^{AB}_{ik}|\sqrt{ \sum_{q=1}^Q \alpha_q^B \left(x_{q,k}^B\right)^2 }
}{
\sum\limits_{l \in U^B_i} |\theta^{AB}_{il}| \sqrt{ \sum_{q=1}^Q \alpha_q^B \left(x_{q,l}^B\right)^2 }
}.
\]

\end{proposition}

\section{Intermediate and second-stage determinantal sampling designs}
\subsection{Determinantal sampling designs}\label{DSD}
A random sample $\mathbb{S}$ from a finite population $U = \{1, \dots, N\}$, governed by a probability law $\mathcal{P}$, is said to follow a \emph{determinantal sampling design (DSD)} if there exists a Hermitian contracting matrix $\mathbf{\mathbf{K}} \in \mathbb{C}^{N \times N}$, that is, a Hermitian matrix whose eigenvalues lie in $[0,1]$, such that for every subset $s \subseteq U$, the probability that $s$ is included in the sample satisfies:
\[
\mathbb{P}(s \subseteq \mathbb{S}) = \det(K_{|s}).
\]

From this definition we directly obtain the first- and second-order inclusion probabilities:
\begin{equation}
\left\{
\begin{array}{l}
\pi_i = \mathbb{P}(\{i\} \subseteq \mathbb{S}) = \det(K_{|\{i\}}) = K_{ii}, \\[0.8em]
\pi_{ij} = \mathbb{P}(\{i, j\} \subseteq \mathbb{S}) = \det(K_{|\{i,j\}}) = K_{ii}K_{jj} - K_{ij}K_{ji} = K_{jj} - K_{ij}\overline{K_{ij}} = K_{ii}K_{jj} - |K_{ij}|^2,
\end{array}
\right.
\label{InclusionProbabilities}
\end{equation}
where $\overline{z}$ denotes the complex conjugate of $z \in \mathbb{C}$ and $|z|$ its modulus.

\vspace{0.5em}

In the case of determinantal sampling designs, the $(N \times N)$ inclusion covariance matrix $\mathbf{\Delta}$, whose entries are defined as $
\mathbf{\Delta}_{ij} = \pi_{ij} - \pi_i \pi_j,
$ admits the closed-form expression:
\begin{equation}
\mathbf{\Delta} = \mathbf{K} \odot (\mathbf{I}_N - \overline{\mathbf{K}}).
\end{equation}
\paragraph{}
The variance of the Horvitz--Thompson estimator under a determinantal sampling design follows from the classical expression (see \citet{sarndal2003model}):
\begin{equation}\label{Vardet}
\operatorname{var}(\hat{t}_\mathbf{y}^{\mathrm{HT}}) = 
\mathbf{y}^\intercal (\mathbf{I}_N \odot \mathbf{K})^{-1} \left[ \mathbf{K} \odot (\mathbf{I}_N - \overline{\mathbf{K}}) \right] (\mathbf{I}_N \odot \mathbf{K})^{-1}\mathbf{y}.
\end{equation}

\paragraph{}
One of the main challenges in applying determinantal sampling designs lies in constructing kernels that satisfy practical constraints commonly encountered in survey sampling, notably fixed sample size and prescribed first-order inclusion probabilities. \citet{loonis2019determinantal} and \citet{loonis2023} show that the fixed-size requirement can be fulfilled by using orthogonal projection kernels. Furthermore, as implied by Equation~\eqref{InclusionProbabilities}, the specification of inclusion probabilities can be achieved by constructing kernels with a prescribed diagonal, i.e., $K_{ii} = \Pi_i$, where $\mathbf{\Pi} = (\Pi_i)_{i=1}^N$ denotes a given vector of desired first-order inclusion probabilities. The authors provide explicit examples of such kernels, including orthogonal projection matrices, and present an algorithm, relying on \citet{fickus2013constructing}, to construct all Hermitian contracting matrices with a given diagonal. They also show how to construct sampling designs that are optimal for some optimisation criterion. More practically, the authors observe that the simplest determinantal designs can be mobilized for populations of several thousand individuals. The size of the
population that is compatible with more sophisticated constructions, particularly those associated with optimization issues, is several hundred. Although not large, such sizes are commonly encountered when managing stratified sampling designs.
\paragraph{}
We conclude this section with some additional properties that will prove useful in the sequel. 
\begin{itemize}
\item(Poisson sampling) Poisson sampling design (with first order inclusion probabilities $\Pr(\{i\} \subseteq \mathbb{S}) = \Pi_i$ belongs to the family of determinantal sampling designs, with kernel $\mathbf{K}=\mathbf{D}_{\mathbf{\Pi}}$;
\item (Complementary Sample). Let $\mathbb{S} \sim DSD(\mathbf{K})$. The complementary sample $\SS^c$ is a determinantal random sample with kernel $\mathbf{I}_N-\mathbf{K}$;
\item (Domain). Let $DSD(\mathbf{K})$ be a determinantal sampling design on $U$ with kernel $\mathbf{K}$, and let $A\subseteq U$ be a subpopulation (or domain). Then the the restriction $DSD(\mathbf{K})_{|A}$ of $DSD(\mathbf{K})$ to $A$ is a determinantal sampling design with kernel $K_{|A}$, the submatrix of $\mathbf{K}$ whose rows and columns are indexed by $A$: $DSD(\mathbf{K})_{|A}=DSD(K_{|A})$;
\item (Stratification). Let $(U_1,\cdots,U_{H})$ be a partition of $U$ into $H$ strata. The sampling design $DSD(\mathbf{K})$ is stratified iff the matrix admits a block diagonal decomposition relative to these strata, that is $i\in U_h,j\in U_{h'}$, $h \neq h'$ implies $K_{kl}=0$. 
\item (Unitary transformation) Let \( DSD(\mathbf{K}) \) denote a determinantal sampling design with kernel \( \mathbf{K} \), and let \( \boldsymbol{\mathcal{U}} \) be a unitary matrix. Then, the matrix \( \boldsymbol{\mathcal{U}} \mathbf{K} \overline{\boldsymbol{\mathcal{U}}}^\intercal \) also defines a determinantal sampling design. Moreover, if \( \boldsymbol{\mathcal{U}} \) is the permutation matrix associated with a permutation \( \sigma \), then  \( \boldsymbol{\mathcal{U}} \mathbf{K} \overline{\boldsymbol{\mathcal{U}}}^\intercal \) defines, over the permuted population \( \sigma(U) \), the same sampling design as \( \mathbf{K} \) does over \( U \).
\end{itemize}

\subsection{Inclusion probabilities in the target population}
Apart from the case where \(\mathcal{P}^A\) corresponds to Simple Random Sampling or Poisson sampling, the target inclusion probabilities \(\pi_k^{1B} = \Pr(\{k\} \subseteq \mathbb{S}^{1B})\) and \(\pi_k^{2B} = \Pr(\{k\} \subseteq \mathbb{S}^{2B})\) are generally difficult or even impossible to compute in practice \citep{deville2006indirect}. When both \(\mathcal{P}^A\) and \(\mathcal{P}^{AB}\) are determinantal sampling designs, the following proposition provides closed-form expressions not only for these marginal probabilities but also for the joint inclusion probabilities.

\begin{proposition}\label{Probacible}
Let \(U^A\), \(U^B\), \(\mathbf{L}^{\!\!AB}\), \(\mathcal{P}^A\), and \(\mathcal{P}^{AB}\) be given, and let \(\mathbb{S}^{1B}\), \(\mathbb{S}^{2B}\) be the resulting one-stage and two-stage indirect random samples drawn from \(U^B\).

Denote by \(N^A_k = |U^A_k|\), \(N^A_l = |U^A_l|\), and \(N^A_{kl} = |U^A_k \cup U^A_l|\) the cardinalities of the subsets \(U^A_k\), \(U^A_l\), and their union respectively. The pair \((U^A_k, k)\) represents the list of pairs \((i, k)\) with \(i \in U^A_k\). This notation similarly applies to \((U^A_l, l)\), \((U^A_k \cup U^A_l, k)\), and \((U^A_k \cup U^A_l, l)\).

Assume that \(\mathcal{P}^A\) is determinantal with kernel $\mathbf{K}^A$ and \(\mathcal{P}^{AB}\) satisties \ref{H3-4}, then \(\mathcal{P}^{AB}\) is determinantal with kernel $\mathbf{K}^{AB}$, and :
\[
\begin{aligned}
\pi^{2B}_k &= \Pr\left(\{k\} \subseteq \mathbb{S}^{2B}\right) = 1 - \det\left(\mathbf{I}_{N^A_k} - \mathbf{K}^A_{|U^A_k} \, \mathbf{K}^{AB}_{|(U^A_k, k)}\right) \\
\pi^{2B}_{k \ell} &= \Pr(\{k, \ell\} \subseteq \mathbb{S}^{2B}) \\
&= 1 + \det \left( \mathbf{I}_{N^A_{k \ell}} - \mathbf{K}^A_{|U^A_k \cup U^A_\ell} \left( \mathbf{K}^{AB}_{|(U^A_k \cup U^A_\ell, k)} + \mathbf{K}^{AB}_{|(U^A_k \cup U^A_\ell, \ell)} \right) \right) \\
&\quad - \det \left( \mathbf{I}_{N^A_k} - \mathbf{K}^A_{|U^A_k} \mathbf{K}^{AB}_{|(U^A_k, k)} \right) - \det \left( \mathbf{I}_{N^A_\ell} - \mathbf{K}^A_{|U^A_\ell} \mathbf{K}^{AB}_{|(U^A_\ell, \ell)} \right).
\end{aligned}
\]
In the context of one-stage indirect sampling, the first-order and joint inclusion probabilities are given by:
\[
\begin{aligned}
\pi^{1B}_k &= \Pr(\{k\} \subseteq \mathbb{S}^{1B}) = 1 - \det \left( \mathbf{I}_{N^A_k} - \mathbf{K}^A_{|U^A_k} \right), \\
\pi^{1B}_{k\ell} &= \Pr(\{k, \ell\} \subseteq \mathbb{S}^{1B}) \\
&= 1 + \det \left( \mathbf{I}_{N^A_{k\ell}} - \mathbf{K}^A_{|U^A_k \cup U^A_\ell} \right) - \det \left( \mathbf{I}_{N^A_k} - \mathbf{K}^A_{|U^A_k} \right) - \det \left( \mathbf{I}_{N^A_\ell} - \mathbf{K}^A_{|U^A_\ell} \right).
\end{aligned}
\]

\end{proposition}

The proof is provided in Appendix \ref{ProofProbacible}. From Proposition \ref{Probacible}, we deduce that computing the various first-order inclusion probabilities requires only the knowledge of $U^A_k$ for each unit $k$. Consequently, the target Horvitz–Thompson estimator, defined as $\hat{t}^{HT}_{\mathbf{y}^B} = \sum_{k \in \mathbb{S}^{2B}} y^B_k/\pi^{2B}_k$, can be applied under the same conditions as the GWSM estimators, as outlined in Figure \ref{Fig2b}

\paragraph{}
When \(\mathbf{L}^{\!\!AB}\) is fully known, both \(\pi^{2B}_k\) and \(\pi^{2B}_{kl}\) depend on \(\mathbf{K}^A\). Consequently, it is possible to obtain different sets of target inclusion probabilities by adjusting this parameter. For example, one can identify a matrix \(\mathbf{K}^A\) that, to the extent possible, satisfies \(\pi^{2B}_k(\mathbf{K}^A) \approx \Pi^B_k\), where \(\Pi^B = (\Pi^B_1, \cdots, \Pi^B_{N^B})^\intercal\) is a specified vector of target inclusion probabilities. Under the assumptions in (\ref{H3-4}), the \(\pi^{2B}_k\) can also be tuned via \(\mathbf{K}^{AB}\), that is, through \(\mathbf{\Pi}^{1AB}\).

\subsection{Optimal intermediate determinantal sampling design}
In the previous sections, we optimised \(\smash{\sum_{q=1}^Q \alpha_q^B \mathrm{var}(\hat{t}_{\mathbf{x}_q^B})}\) over \(\mathbf{\Theta}^{\!AB}\) and \(\mathbf{\Pi}^{1AB}\). According to Proposition \ref{variance}, this objective function also depends on \(\tilde{\mathbf{\Delta}}^A\), which is itself a function of \(\mathbf{K}^A\) when \(\mathcal{P}^A\) is determinantal:
\[
\smash{\tilde{\mathbf{\Delta}}^A = (\mathbf{I}_{N^A} \odot \mathbf{K}^A)^{-1} \left[ \mathbf{K}^A \odot (\mathbf{I}_{N^A} - \overline{\mathbf{K}^A}) \right] (\mathbf{I}_{N^A} \odot \mathbf{K}^A)^{-1}}.
\]
Therefore, for fixed \(\mathbf{\Theta}^{\!AB}\) and \(\mathcal{P}^{AB}\), one can search for the kernel \(\mathbf{K}^A\) that solves the following optimisation problem:
\begin{equation}
\min_{\mathbf{K}^A} \sum_{q=1}^Q \alpha_q^B \mathrm{var}(\hat{t}_{\mathbf{x}_q^B}) \quad \text{subject to} \quad
\left\{
\begin{array}{l}
\mathbf{K}^A = \overline{\mathbf{K}^A}^\intercal, \\
\mathbf{K}^A \mathbf{K}^A = \mathbf{K}^A, \\
\mathrm{diag}(\mathbf{K}^A) = \Pi^A.
\end{array}
\right.
\end{equation}
These constraints ensure that \(\mathbf{K}^A\) is an orthogonal projection matrix with a prescribed diagonal, yielding an optimal fixed-size intermediate determinantal sampling design with given inclusion probabilities.

\paragraph{}
Given that the first-order and joint inclusion probabilities are explicitly known, it is also possible to employ the optimal target Horvitz-Thompson estimator, minimising instead, under the same set of constraints:
\[
\smash{\sum_{q=1}^Q \alpha_q^B \mathrm{var}(\hat{t}^{HT}_{\mathbf{x}_q^B})}
\]

\paragraph{}
Furthermore, since the variance of the intermediate Horvitz-Thompson estimators depends on \(\mathbf{K}^A\), one might also aim to enhance estimation accuracy for a set \(\mathbf{X}^A\) of \(P\) auxiliary variables defined on \(U^A\). This leads to composite objective functions, such as:
\[
\smash{\sum_{p=1}^P \alpha_p^A \mathrm{var}(\hat{t}^{HT}_{\mathbf{x}_p^A}) + \sum_{q=1}^Q \alpha_q^B \mathrm{var}(\hat{t}_{\mathbf{x}_q^B})}
\quad \text{or} \quad
\smash{\sum_{p=1}^P \alpha_p^A \mathrm{var}(\hat{t}^{HT}_{\mathbf{x}_p^A}) + \sum_{q=1}^Q \alpha_q^B \mathrm{var}(\hat{t}^{HT}_{\mathbf{x}_q^B})}.
\]
In the next section, we will discuss practical methods to solve these optimisation problems.

\section{Application study}
\subsection{Global setting
}In this section \(U^A\) denotes a set of Primary Units (PUs) within the Normandy region of France. A random sample \(\mathbb{S}^A\) is drawn from \(U^A\) to efficiently coordinate a network of surveyors conducting face-to-face household interviews. Each selected PU is assigned to a surveyor. Within each chosen PU, a random sample of households is subsequently selected for the survey. The PUs are organized to minimize travel distances between households while ensuring a balanced workload for surveyors. In this section, we focus exclusively on the selection of the PUs. \(U^A\) is primarily used for most surveys. However, the characteristics of its units may not always meet the requirements of certain specific face-to-face household surveys. Therefore, an alternative set of Primary Units, \(U^B\), is introduced, from which a second sample \(\mathbb{S}^{2B}\) has to be drawn. Since the same network of surveyors is employed for both types of surveys, it is crucial that the two samples are approximately the same size, and that each selected unit from \(U^B\) is geographically close to a unit from \(U^A\) so that they can be assigned to the same surveyor. To maintain the quality of the estimates, it is also important that both samples — drawn from \(U^A\) and \(U^B\) respectively — are balanced with respect to specific sets of auxiliary variables \(\mathbf{X}^A\) and \(\mathbf{X}^B\).

 \paragraph{}
We especially build two specific sets $U^A$ and $U^B$. Their respective sizes are $N^A=250$ and $N^B=337$ (see Figure \ref{Fig4}). 

\begin{figure}[h!]
\begin{subfigure}[t]{0.47\textwidth}
\centering
\includegraphics[scale=1.4]{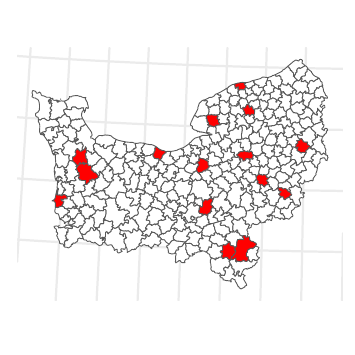}\vfill
    \caption{$U^A$ of size $N^A=250$}
    \end{subfigure}
    \begin{subfigure}[t]{0.47\textwidth}
\centering
\includegraphics[scale=1.4]{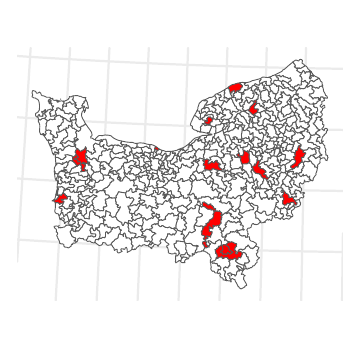}\vfill
    \caption{$U^B$ of size $N^B=337$}
    \end{subfigure}
    \caption{One unit of $\SS^A$ is close to one of $\SS^{2B}$, while being almost of the same number and, to the extend possible, $\SS^A$ and $\SS^B$ are balanced with respect to specific sets of auxiliary variables $\mathbf{X}^A$ and $\mathbf{X}^B$.}\label{Fig4}
\end{figure}
\paragraph{}
To achieve our goal to implement the aforementioned methodology, we construct a specific link matrix \(\mathbf{L}^{\!\!AB}\). For each unit \(i \in U^A\), we first consider all units in \(U^B\), whose territories overlap with that of \(i\), as linked to \(i\). As noted in Section \ref{Sec_Total}, the resulting number of links is too large to manage the associated matrices effectively. Therefore, among the units of \(U^B\) overlapping with unit \(i \in U^A\), only the two units sharing the smallest population with \(i\) are retained as linked (see Figure \ref{Fig5}). 
This procedure leaves some units of \(U^B\) without any links; these cases are treated separately. The final number of links amounts to 631.

\begin{figure}[h!]
\centering
\includegraphics[scale=0.05,trim=6pt 45pt 6pt 45pt, clip, width=0.5\textwidth]{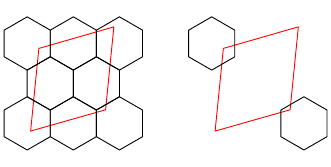}
    \caption{\raggedright Construction of $\mathbf{L}^{\!\!AB}$ to obtain a \textit{reasonable} number of links: Among the units of $U^B$ whose territory overlaps that of unit $i$ of $U^A$, only the two units sharing the smallest number on inhabitants with $i$ are considered linked to this unit.}\label{Fig5}
\end{figure}
\paragraph{}
For each population $U^A$ and $U^B$ a sampling frame is available. $\mathbf{X}^A$ and $\mathbf{X}^B$ consist of two sets of $P=Q=3$ auxiliary variables:
\begin{enumerate}
\item the total number of households, whose \textit{head of household} is aged 55 and more,
\item the total amount of retirement pensions,
\item the total amount of work incomes.
\end{enumerate}

We also consider $\mathbf{y}^A$ and $\mathbf{y}^B$  as variables of interest, representing the total amount of unemployment income. While these variables are not included in the optimization process, we assess the impact of our approach on the accuracy of their total estimates across the entire population
Finally, the intermediate first-order inclusion probabilities are prescribed. The $\Pi^A_k$ are proportional to the number of inhabitants with a sampling size fixed to $n^A=15$. 

\subsection{Intermediate and target balanced samplings}

To obtain intermediate and target balanced samplings, we simultaneously consider the three above mentionned  optimisation problems over $\mathbf{\Theta}^{\!AB}$, $\mathbf{\Pi}^{1AB}$ and $\mathbf{K}^A$. This global problem can be expressed as: 

\begin{equation}
\underset{\mathbf{K}^A,\mathbf{\Theta}^{\!AB},\mathbf{\Pi}^{1AB}}{Min}\left [\sum_{p=1}^P \alpha^A_p var(\hat{t}^{HT}_{\mathbf{x}^A_p})+ \sum_{q=1}^Q \alpha_q^B var(\hat{t}_{\mathbf{x}^B_q}) \right ]\text{ s.c. }\left \{ 
\begin{array}{l}
{\mathbf{K}^A}^\intercal=\mathbf{K}^A,\\
\mathbf{K}^A\mathbf{K}^A=\mathbf{K}^A,\\
diag(\mathbf{K}^A)=\Pi^A,\\
(\mathbf{L}^{\!\!AB}\odot \mathbf{\Theta}^{\!AB})^\intercal \uvec{1}_{N^A}=\uvec{1}_{N^B},\\
(\mathbf{L}^{\!\!AB}\odot\mathbf{\Pi}^{1AB}_1)\uvec{1}_{N^B}=\uvec{1}_{N^A}.
\end{array}
\right.
\end{equation}

The previous optimisation process is restricted to real kernels, which leads to the simplification $\overline{\mathbf{K}}^{A^\intercal}=\mathbf{K}^A \Leftrightarrow {\mathbf{K}^A}^\intercal=\mathbf{K}^A$. The constraint $(\mathbf{L}^{\!\!AB}\odot\mathbf{\Pi}^{1AB}1)\uvec{1}_{N^B}=\uvec{1}_{N^A}$ provides a matrix formulation for the requirement that each random set $\mathbb{S}_i$ has a fixed cardinality of 1, as stated in Assumptions \ref{H3-4}. The parameters are set as $\alpha^A_p=1/(t_{\mathbf{x}^A_p})^2$ and $\alpha^B_q=1/(t_{\mathbf{x}^B_q})^2$ for $p, q = 1, 2, 3$, ensuring that the objective function incorporates the squared coefficient of variation for each estimator:
$$\sum_{p=1}^P \operatorname{cv}^2(\hat{t}^{HT}_{\mathbf{x}^A_p})+ \sum_{q=1}^Q \operatorname{cv}^2(\hat{t}_{\mathbf{x}^B_q}). $$

We also consider initial values for each parameter:
\begin{enumerate}
\item Let $\mathbf{K}^A_0 = \mathbf{P}^{\mathbf{\Pi}^A}$, where $\mathbf{P}^\mathbf{\Pi}$ is the family of projection matrices with diagonal $\mathbf{\Pi}$, as introduced by \citet{loonis2019determinantal} and further analyzed by \citet{loonis2023}.
\item Let $\mathbf{\Theta}^{\!AB}_0$ be such that its non-zero entries are constant for each unit in $U^A_k$; in other words, except for the entries that are equal to 0, all other entries are constant within each column, with each column summing to 1.
\item Let $\mathbf{\Pi}^{1AB}_0$ be such that its non-zero entries are constant for each unit of $U^B_i$; in other words, apart from the entries that are equal to 0, all other entries are constant within each row, with each row summing to 1. The initial sampling designs $\mathcal{P}^B_i$ have uniform probabilities.
\end{enumerate}
\paragraph{}
The optimisation process consists in the following Coordinate descent Algorithm:

\begin{algorithm}[H]
\caption{Alternating optimization of \( C_{\mathbf{X}^P,\mathbf{X}^Q}(\mathbf{K}^A,\mathbf{\Theta}^{\!AB},\mathbf{\Pi}^{1AB}) = \sum_{p=1}^P \operatorname{cv}^2(\hat{t}^{HT}_{\mathbf{x}^A_p}) + \sum_{q=1}^Q \operatorname{cv}^2(\hat{t}_{\mathbf{x}^B_q}) \)}
\label{algo:alternating-optimization}
\begin{algorithmic}[1]
\Require Initial values: \( \mathbf{K}^A_0 \), \( \mathbf{\Theta}^{\!AB}_0 \), \( \mathbf{\Pi}^{1AB}_0 \); number of iterations \( R \)
\vspace{0.5em}
\Statex \textbf{Repeat for} \( r = 1 \) \textbf{to} \( R \), optimise $C(...)$ over:
    \State \( \mathbf{K}^A \) using Algorithm~\ref{algo:rotation-optimization} (see Section~\ref{Greedy}), with \( \mathbf{\Theta}^{\!AB} = \mathbf{\Theta}^{\!AB}_0 \), \( \mathbf{\Pi}^{1AB} = \mathbf{\Pi}^{1AB}_0 \)
    \Statex \quad Let \( \mathbf{K}^A_1 \) be the result
    \State \( \mathbf{\Theta}^{\!AB} \) using Proposition~\ref{PropOpt}, with \( \mathbf{K}^A = \mathbf{K}^A_1 \), \( \mathbf{\Pi}^{1AB} = \mathbf{\Pi}^{1AB}_0 \)
    \Statex \quad Let \( \mathbf{\Theta}^{\!AB}_1 \) be the result
    \State \( \mathbf{\Pi}^{1AB} \) using Proposition~\ref{PiOpt}, with \( \mathbf{\Theta}^{\!AB} = \mathbf{\Theta}^{\!AB}_1 \)
    \Statex \quad Let \( \mathbf{\Pi}^{1AB}_1 \) be the result
    \Statex Update: \( \mathbf{K}^A_0 \gets \mathbf{K}^A_1 \), \( \mathbf{\Theta}^{\!AB}_0 \gets \mathbf{\Theta}^{\!AB}_1 \), \( \mathbf{\Pi}^{1AB}_0 \gets \mathbf{\Pi}^{1AB}_1 \)
\end{algorithmic}
\end{algorithm}

\subsection{Results}

The figures of this section are provided in Appendix \ref{Figures}. 
\paragraph{}
Figure \ref{Fig_opt} demonstrates the effectiveness of our optimisation process, which consistently yields more accurate estimates for the target auxiliary variables, while enhancing the precision of the intermediate estimates. 

\begin{itemize}
\item The initial value of the objective function is $8.2$. After $R=5$ iteration and $15$ steps, this value is reduced by a factor of 10, reaching a final value of $0.82$. 
\item The reduction is particularly pronounced for the \textit{target} contribution ($\sum_{q=1}^Q \operatorname{cv}^2(\hat{t}_{\mathbf{x}^B_q})$), which decreases by a factor of 12.6 (from 7.3 to 0.58), whereas the intermediate contribution ($\smash{\sum_{q=1}^Q \operatorname{cv}^2(\hat{t}^{HT}_{\mathbf{x}^A_q})}$) is reduced by a factor of 3.8 (from 0.9 to 0.25).
\end{itemize}

Figure \ref{samplesize} illustrates the variability of the target sample size of $\mathbb{S}^{2B}$. With a fixed intermediate sampling size of $n^A = 15$, the target sample size is also 15 in 90 percent of the cases. 
\paragraph{}Figure \ref{Fig_opt_interet} illustrates the effectiveness of our procedure for the variables of interest, $\mathbf{y}^A$ and $\mathbf{y}^B$, which are not involved in the optimization process but are correlated with the auxiliary information. For the intermediate population, the Horvitz-Thompson estimator is used. For the target population, the optimization procedure leads naturally to the use of the GWSM estimator. At each iteration step, Proposition \ref{Probacible} also allows the use of the Horvitz-Thompson estimator, whose precision is likewise computed. 

\paragraph{}For the target population, the coefficient of variation of the GWSM estimator is reduced by a factor of 3.6, and its squared value by a factor of 13.4—matching the gain observed for the auxiliary variables. Using the Horvitz-Thompson estimator for the target population also yields a reduction in the coefficient of variation by a factor of 1.8, and in its square by a factor of 3.2. This reduction is, as expected, less pronounced.  
\paragraph{}
To conclude this section, we note that one could, in principle, attempt to optimize the criterion involving the target Horvitz-Thompson estimator using Proposition \ref{Probacible}:
$$\sum_{p=1}^P \operatorname{cv}^2(\hat{t}^{HT}_{\mathbf{x}^A_p})+ \sum_{q=1}^Q \operatorname{cv}^2(\hat{t}^{HT}_{\mathbf{x}^B_q}). $$

However, computing the full set of target joint inclusion probabilities is computationally prohibitive, making this approach less practical. This limitation highlights a promising avenue for future research.

\section{Conclusion}In this article, we leverage the versatility of determinantal sampling designs to address a practical problem requiring knowledge of high-order inclusion probabilities within the framework of two-stage indirect sampling. We begin by proposing three extensions to this framework (Propositions \ref{variance}, \ref{PropOpt}, and \ref{PiOpt}). Furthermore, when both the intermediate and second-stage sampling designs are determinantal, we show that the target inclusion probabilities can be expressed in closed form. The optimal GWSM weight matrix can not only be computed in practice, but also further optimised.
\paragraph{}
This setting gives rise to several challenges, such as improving optimisation procedures to efficiently handle the optimised target Horvitz–Thompson estimator,  identifying intermediate and second-stage determinantal designs that minimise the variance of the target sample size, looking for close-from formulae for the target inclusion probabilities for general second-stage determinantal sampling designs or searching for connections with the double weight-share method as discussed in \citet{medous2025optimal}. The observed flexibility of the framework also opens the door to addressing other problems, such as identifying the most spatially balanced sampling design, exploring alternatives to the Horvitz–Thompson estimator, as discussed in \citet{theberge2017estimation}. 
\paragraph{}
Beyond these applications, new opportunities has recently emerged with the introduction by \citet{panahbehagh2023geometric} of a novel family of sampling designs, whose inclusion probabilities are also explicitly known and parametrisable at any order. Comparing the properties and versatility of these two families of designs appears to be a promising direction for future research.

\section*{Acknowledgments}
 The views expressed in this article are those of the author and do not necessarily represent those of his affiliated institutions. 
The author thanks Olivier Guin for the relevance of his comments, which led to substantial improvements in this version of the article. The author is also grateful to the participants of the internal seminar of the Department of Statistical Methods at Insee for their valuable feedback and constructive suggestions. Any remaining errors, omissions, or inaccuracies are the sole responsibility of the author.

\bibliographystyle{apalike}
\bibliography{biblio}

\appendix

\section{Proof of Proposition \ref{variance}}\label{Proofvariance}
We provide the proof for the particular case $N^A=3$ and $N^B=2$. We write $$\operatorname{var}(\hat{t}_{\mathbf{y}^B})=\sum_{k\in U^B}\sum_{l\in U^B} y_k^By_l^B \mathbb{E}(w_kw_l)-(t_{\mathbf{y}^B})^2$$
$$y^B_ky^B_l\mathbb{E}(w_kw_l)=\sum_{i\in U^A}\sum_{j\in U^A}\frac{y_k^By_l^B\ell^{AB}_{ik}\theta^{AB}_{ik}\ell^{AB}_{jl}\theta^{AB}_{jl}\overset{\pi^{A}_{ij}}{\overbrace{\mathbb{E}(1(i\in S^A)1(j\in S^A))}}\overset{\pi^{AB}_{(i,k)(j,l)}}{\overbrace{\mathbb{E}(1(k\in S_i^A)1(l\in S_j^A)) }}}{\pi^{A}_i\pi^{A}_j\pi^{AB}_{(i,k)}\pi^{AB}_{(j,l)}}=$$ $$
\sum_{i\in U^A}\sum_{j\in U^A}\frac{y_k^By_l^B\ell^{AB}_{ik}\theta^{AB}_{ik}\ell^{AB}_{jl}\theta^{AB}_{jl}\pi^{A}_{ij}\pi^{AB}_{(i,k)(j,l)}}{\pi^{A}_i\pi^{A}_j\pi^{AB}_{(i,k)}\pi^{AB}_{(j,l)}}$$

We look for a matrix $\mathbf{\mathcal{Q}}$ such that $\sum_{k\in U^B}\sum_{l\in U^B} y_k^By_l^B \mathbb{E}(w_kw_l)=\mathbf{v}_{\mathbf{\Theta}^{\!AB}} ^\intercal \mathbf{\mathcal{Q}} \mathbf{v}_{\mathbf{\Theta}^{\!AB}}$, with $\mathbf{\mathcal{Q}}$ function of $\pi^A,\pi^{2A},\mathbf{\Pi}^{1AB},\mathbf{\Pi}^{2AB},\mathbf{y}^B,\mathbf{L}^{\!\!AB}$. We observe that:

\[ \begin{array}{c}
\sum_{k\in U^B}\sum_{l\in U^B}\sum_{i\in U^A}\sum_{j\in U^A}y_k^By_l^B\theta^{AB}_{ik}\theta^{AB}_{jl}\ell_{ik}\ell_{jl}\pi^{A}_{ij}\pi^{AB}_{(i,k)(j,l)}=\\
\mathbf{v}_{\mathbf{\Theta}^{\!AB}}^\intercal \mathbf{D}_{\mathbf{v}_{L_{AB}}} \\
\left (
\setlength{\arraycolsep}{2pt} % Réduit l'espacement entre les colonnes
\begin{array}{@{}cccccc@{}}
y_1^By_1 ^B\pi^{A}_{11 }\pi^{AB}_{(1,1)(1,1)}&
y_1 ^By_1 ^B\pi^{A}_{12}\pi^{AB}_{(1,1)(2,1)}&
y_1 ^By_1 ^B\pi^{A}_{13}\pi^{AB}_{(1,1)(3,1)}&
y_1 ^By_2 ^B\pi^{A}_{11 }\pi^{AB}_{(1,1)(1,2)}&
y_1 ^By_2 ^B\pi^{A}_{12}\pi^{AB}_{(1,1)(2,2)}&
y_1 ^By_2 ^B\pi^{A}_{13}\pi^{AB}_{(1,1)(3,2)}
\\
y_1 ^By_1 ^B\pi^{A}_{21 }\pi^{AB}_{(2,1)(1,1)}&
y_1 ^By_1 ^B\pi^{A}_{22}\pi^{AB}_{(2,1)(2,1)}&
y_1 ^By_1 ^B\pi^{A}_{23}\pi^{AB}_{(2,1)(3,1)}&
y_1 ^By_2 ^B\pi^{A}_{21 }\pi^{AB}_{(2,1)(1,2)}&
y_1 ^By_2 ^B\pi^{A}_{22}\pi^{AB}_{(2,1)(2,2)}&
y_1 ^By_2 ^B\pi^{A}_{23}\pi^{AB}_{(2,1)(3,2)}
\\ 
y_1 ^By_1 ^B\pi^{A}_{31 }\pi^{AB}_{(3,1)(1,1)}&
y_1 ^By_1 ^B\pi^{A}_{32}\pi^{AB}_{(3,1)(2,1)}&
y_1 ^By_1 ^B\pi^{A}_{33}\pi^{AB}_{(3,1)(3,1)}&
y_1 ^By_2 ^B\pi^{A}_{31 }\pi^{AB}_{(3,1)(1,2)}&
y_1 ^By_2 ^B\pi^{A}_{32}\pi^{AB}_{(3,1)(2,2)}&
y_1 ^By_2 ^B\pi^{A}_{33}\pi^{AB}_{(3,1)(3,2)}
\\ 
y_2^By_1 ^B\pi^{A}_{11 }\pi^{AB}_{(1,2)(1,1)}&
y_2^By_1 ^B\pi^{A}_{12}\pi^{AB}_{(1,2)(2,1)}&
y_2^By_1 ^B\pi^{A}_{13}\pi^{AB}_{(1,2)(3,1)}&
y_2^By_2 ^B\pi^{A}_{11 }\pi^{AB}_{(1,2)(1,2)}&
y_2^By_2 ^B\pi^{A}_{12}\pi^{AB}_{(1,2)(2,2)}&
y_2^By_2 ^B\pi^{A}_{13}\pi^{AB}_{(1,2)(3,2)}
\\
y_2^By_1 ^B\pi^{A}_{21 }\pi^{AB}_{(2,2)(1,1)}&
y_2^By_1 ^B\pi^{A}_{22}\pi^{AB}_{(2,2)(2,1)}&
y_2^By_1 ^B\pi^{A}_{23}\pi^{AB}_{(2,2)(3,1)}&
y_2^By_2 ^B\pi^{A}_{21 }\pi^{AB}_{(2,2)(1,2)}&
y_2^By_2 ^B\pi^{A}_{22}\pi^{AB}_{(2,2)(2,2)}&
y_2^By_2 ^B\pi^{A}_{23}\pi^{AB}_{(2,2)(3,2)}
\\
y_2^By_1 ^B\pi^{A}_{31 }\pi^{AB}_{(3,2)(1,1)}&
y_2^By_1 ^B\pi^{A}_{32}\pi^{AB}_{(3,2)(2,1)}&
y_2^By_1 ^B\pi^{A}_{33}\pi^{AB}_{(3,2)(3,1)}&
y_2^By_2 ^B\pi^{A}_{31 }\pi^{AB}_{(3,2)(1,2)}&
y_2^By_2 ^B\pi^{A}_{32}\pi^{AB}_{(3,2)(2,2)}&
y_2^By_2 ^B\pi^{A}_{33}\pi^{AB}_{(3,2)(3,2)}
\end{array}
\right ) \\ \mathbf{D}_{\mathbf{v}_{L_{AB}}}\mathbf{v}_{\mathbf{\Theta}^{\!AB}} \\ 
=\mathbf{v}_{\mathbf{\Theta}^{\!AB}}^\intercal \mathbf{D}_{\mathbf{v}_{L_{AB}}}\left [((\mathbf{y}^B\mathbf{y}^{B^\intercal})\otimes\mathbf{J}_{N^A})\odot(\mathbf{J}_{N^B}\otimes \Pi^{2A})\odot \mathbf{\Pi}^{2AB}\right ]\mathbf{D}_{\mathbf{v}_{L_{AB}}}\mathbf{v}_{\mathbf{\Theta}^{\!AB}}.
\end{array}
\]
As a consequence
$\sum_{k\in U^B}\sum_{l\in U^B} y_k^By_l^B \mathbb{E}(w_kw_l)=
\sum_{k\in U^B}\sum_{l\in U^B}\sum_{i\in U^A}\sum_{j\in U^A}\frac{y_k^By_l^B\theta^{AB}_{ik}\theta^{AB}_{jl}\ell^{AB}_{ik}\ell^{AB}_{jl}\pi^{A}_{ij}\pi^{AB}_{(i,k)(j,l)}}{\pi^{A}_i\pi^{A}_j\pi^{AB}_{(i,k)}\pi^{AB}_{(j,l)}}=$ $$
\mathbf{v}_{\mathbf{\Theta}^{\!AB}}^\intercal \mathbf{D}_{\mathbf{v}_{L_{AB}}}\left [((\mathbf{y}^B\mathbf{y}^{B^\intercal})\otimes\mathbf{J}_{N^A})\odot(\mathbf{J}_{N^B}\otimes (D^{-1}_{\Pi^A}\Pi^{2A}D^{-1}_{\Pi^A}))\odot \left(D^{-1}_{\mathbf{v}_{\mathbf{\Pi}^{1AB}}}\mathbf{\Pi}^{2AB}D^{-1}_{\mathbf{v}_{\mathbf{\Pi}^{1AB}}}\right)\right ]\mathbf{D}_{\mathbf{v}_{L_{AB}}}\mathbf{v}_{\mathbf{\Theta}^{\!AB}},$$ where $\Pi^{2A}$ is the $(N^A\times N^A)$ matrix such that $\Pi^{2A}_{ij}=\pi^A_{ij}$.  To conclude we observe that $D^{-1}_{\Pi^A}\Pi^{2A}D^{-1}_{\Pi^A}=\tilde{\mathbf{\Delta}}^A+\mathbf{J}_{N^A}$ and $D^{-1}_{\mathbf{v}_{\mathbf{\Pi}^{1AB}}}\mathbf{\Pi}^{2AB}D^{-1}_{\mathbf{v}_{\mathbf{\Pi}^{1AB}}}=\tilde{\mathbf{\Delta}}^{AB}+\mathbf{J}_{N^AN^B}$, leading to Proposition \ref{variance}, as the additional terms equal $(t_{\mathbf{y}^B})^2$.

\section{Implementing some matrices}\label{implement}
In this section we showcase how implementing some of the matrices used in the text, when 
$N^A=N^B=3$. We start by the unbiasedness constraint $(\mathbf{L}^{\!\!AB} \odot \mathbf{\Theta}^{\!AB})^\intercal\uvec{1}_{N^A}=\uvec{1}_{N^B}\Leftrightarrow \mathbf{E}^{AB}\mathbf{v}_{\mathbf{\Theta}^{\!AB}}=\uvec{1}_{N^B}$.  

$$ \left \{
\begin{array}{c}
\left( \underbrace{\left( 
\begin{array}{ccc}
     \ell^{AB}_{11} & \ell^{AB}_{12} & \ell^{AB}_{13}\\
     \ell^{AB}_{21} & \ell^{AB}_{22} & \ell^{AB}_{23}\\
     \ell^{AB}_{31} & \ell^{AB}_{32} & \ell^{AB}_{33}\\
\end{array}
\right)}_{\mathbf{L}^{\!\!AB}}
\odot 
\underbrace{\left( 
\begin{array}{ccc}
     \theta^{AB}_{11} & \theta^{AB}_{12}  & \theta^{AB}_{13}\\
     \theta^{AB}_{21} & \theta^{AB}_{22}  & \theta^{AB}_{23}\\
     \theta^{AB}_{31} & \theta^{AB}_{32}  & \theta^{AB}_{33}\\
\end{array}
\right) }_{\mathbf{\Theta}^{\!AB}}
\right)^\intercal
\left( 
\begin{array}{cc}
1 \\
1 \\
1 \\
\end{array}
\right) 
= 
\left( 
\begin{array}{cc}
1 \\
1 \\
1\\
\end{array}
\right)\\
\Updownarrow \\

\underbrace{\left( 
\begin{array}{ccccccccc}
    \ell^{AB}_{11} & \ell^{AB}_{21} & \ell^{AB}_{31} & 0 & 0 & 0 & 0 & 0 & 0\\
    0 & 0 & 0 & \ell^{AB}_{12} & \ell^{AB}_{22} & \ell^{AB}_{32} & 0 & 0 & 0 \\
    & 0 & 0 & 0& 0 & 0 & \ell^{AB}_{13} & \ell^{AB}_{23} & \ell^{AB}_{33}
\end{array}
\right)}_{\mathbf{E}^{AB}}
\underbrace{\left( 
\begin{array}{c}
     \theta^{AB}_{11} \\
     \theta^{AB}_{21} \\
     \theta^{AB}_{31} \\
     \theta^{AB}_{12} \\
     \theta^{AB}_{22} \\
     \theta^{AB}_{32} \\
     \theta^{AB}_{13} \\
     \theta^{AB}_{23} \\
     \theta^{AB}_{33} \\
\end{array}
\right)}_{\mathbf{v}_{\mathbf{\Theta}^{\!AB}}}= 
\left( 
\begin{array}{cc}
1 \\
1 \\
1 \\
\end{array}
\right).
\end{array} \right.
$$

We now consider a specific link matrix to show how to implement $\mathbf{\Pi}^{1AB}$ and $\mathbf{\Pi}^{2AB}$.

 $$\mathbf{L}^{\!\!AB}=\left (
\begin{array}{ccc}
   1  & 0 & 1\\
   1  &  1 & 0 \\
   0 & 1 & 1 
\end{array}\right ) \Longrightarrow \mathbf{\Pi}^{1AB}=
\left (
\begin{array}{ccc}
   \pi^{AB}_{11}  & 0 & \pi^{AB}_{13}\\
   \pi^{AB}_{21}  &  \pi^{AB}_{22} & 0 \\
   0 & \pi^{AB}_{32} & \pi^{AB}_{33}\\
\end{array} \right).
$$
When Assumptions \ref{H3-4} hold, one has $\pi^{AB}_{11}+\pi^{AB}_{13}$=$\pi^{AB}_{21}+\pi^{AB}_{22}$=$\pi^{AB}_{32}+\pi^{AB}_{33}=1$.
\paragraph{}
To fill in $\mathbf{\Pi}^{2AB}$, in line with Section \ref{Proofvariance}, one has to consider $\mathbf{\Pi}^{2AB}_{mn} = \pi^{AB}_{(i,k)(j,l)}$ with $1\leq m,n\leq N_AN_B$ and $(i,j,k,l)$ such that $1\leq i,j \leq N^A$, $1\leq k,l \leq N^B$, $m=i + (k-1)N^A$ and $n=j + (l-1)N^A$.
\newline
When Assumptions \ref{H3-4} hold, $\mathbf{\Pi}^{2AB}$ writes:

\[
\begin{blockarray}{cccc | ccccccccc}
 & & & l & 1 & 1 & 1 & 2 & 2 & 2 & 3 & 3 & 3 \\
 & & & j & 1 & 2 & 3 & 1 & 2 & 3 & 1 & 2 & 3 \\ \hline
 & & & n & 1 & 2 & 3 & 4 & 5 & 6 & 7 & 8 & 9 \\ \hline
k & i & m &   &   &   &   &   &   &   &   &   &   \\ \hline
\begin{block}{cccc | ccccccccc}
1 & 1 & 1 &   & \pi^{AB}_{(1,1)} & \pi^{AB}_{(1,1)}\pi^{AB}_{(2,1)} & 0 & 0 & \pi^{AB}_{(1,1)}\pi^{AB}_{(1,2)} & \pi^{AB}_{(1,1)}\pi^{AB}_{(3,2)} & 0 & 0 & \pi^{AB}_{(1,1)}\pi^{AB}_{(3,3)} \\
1 & 2 & 2 &   & \pi^{AB}_{(2,1)}\pi^{AB}_{(1,1)} & \pi^{AB}_{(2,1)} & 0 & 0 & 0 & \pi^{AB}_{(2,1)}\pi^{AB}_{(3,2)} & \pi^{AB}_{(2,1)}\pi^{AB}_{(1,3)} & 0 & \pi^{AB}_{(2,1)}\pi^{AB}_{(3,3)} \\
1 & 3 & 3 &   & 0 & 0 & 0 & 0 & 0 & 0 & 0 & 0 & 0 \\
2 & 1 & 4 &   & 0 & 0 & 0 & 0 & 0 & 0 & 0 & 0 & 0 \\
2 & 2 & 5 &   & \pi^{AB}_{(2,2)}\pi^{AB}_{(1,1)} & 0 & 0 & 0 & \pi^{AB}_{(2,2)} & \pi^{AB}_{(2,2)}\pi^{AB}_{(3,2)} & \pi^{AB}_{(2,2)}\pi^{AB}_{(3,1)} & 0 & \pi^{AB}_{(2,2)}\pi^{AB}_{(3,3)} \\
2 & 3 & 6 &   & \pi^{AB}_{(3,2)}\pi^{AB}_{(1,1)} & \pi^{AB}_{(3,2)}\pi^{AB}_{(2,1)} & 0 & 0 & \pi^{AB}_{(3,2)}\pi^{AB}_{(2,2)} & \pi^{AB}_{(3,2)} & \pi^{AB}_{(3,2)}\pi^{AB}_{(1,3)} & 0 & 0 \\
3 & 1 & 7 &   & 0 & \pi^{AB}_{(1,3)}\pi^{AB}_{(2,1)} & 0 & 0 & \pi^{AB}_{(1,3)}\pi^{AB}_{(2,2)} & \pi^{AB}_{(1,3)}\pi^{AB}_{(3,2)} & \pi^{AB}_{(1,3)} & 0 & \pi^{AB}_{(1,3)}\pi^{AB}_{(3,3)} \\
3 & 2 & 8 &   & 0 & 0 & 0 & 0 & 0 & 0 & 0 & 0 & 0 \\
3 & 3 & 9 &   & \pi^{AB}_{(3,3)}\pi^{AB}_{(1,1)} & \pi^{AB}_{(3,3)}\pi^{AB}_{(2,1)} & 0 & 0 & \pi^{AB}_{(3,3)}\pi^{AB}_{(2,2)} & 0 & \pi^{AB}_{(3,3)}\pi^{AB}_{(1,3)} & 0 & \pi^{AB}_{(3,3)} \\
\end{block}
\end{blockarray}
\]
\section{Proof of Proposition \ref{PiOpt}}\label{ProofPiOpt}

When Assumptions \ref{H3-4} hold, $\tilde{\mathbf{\Delta}}^{AB}$ is a function of $\mathbf{\Pi}^{1AB}$, such that:

\[
\tilde{\Delta}^{AB}_{(i,k),(j,l)} = 
\begin{cases}
\displaystyle\frac{1 - \pi^{AB}_{(i,k)}}{\pi^{AB}_{(i,k)}} & \text{if } i=j, k=l\text{ and }\pi^{AB}_{(i,k)}>0, \\
-1 &  \text{if } i=j\text{ and } (k,l)\in U^B_i\times U^B_i, k\neq l,\\
0 & \text{otherwise.}
\end{cases}
\]
The only terms that contribute to \( \operatorname{var}(\hat{t}_{\mathbf{y}^B}) \) and depend on \( \mathbf{\Pi}^{1AB} \) are those for which \( i = j \), \( k = l \), and \( k \in U^B_i \). According to Section~\ref{Proofvariance}, we therefore only need to consider the following sum in our problem:

$$
\sum_{k\in U^B}\sum_{i\in U^A}\frac{(y_k^B)^2(\theta^{AB}_{ik})^2(\ell^{AB}_{ik})^2\pi^{A}_{i}\pi^{AB}_{(i,k)}}{(\pi^{A}_i\pi^{AB}_{(i,k)})^2}=\sum_{i\in U^A}\sum_{k\in U^B_i}\frac{(y_k^B)^2(\theta^{AB}_{ik})^2}{\pi^{A}_i\pi^{AB}_{(i,k)}}.$$
Minimising $\sum_{q=1}^Q \alpha^B_q \operatorname{var}(\hat{t}_{\mathbf{x}_q^B})$ over $\mathbf{\Pi}^{1AB}$ subject to \ref{H3-4} writes:

$$\underset{\mathbf{\Pi}^{1AB}}{Min}\sum_{i\in U^A}\sum_{k\in U^B_i}\sum_{q=1}^Q\frac{\alpha^B_q(\mathbf{x}_{q,k}^B)^2(\theta^{AB}_{ik})^2}{\pi^{A}_i\pi^{AB}_{(i,k)}}\text{ subject to }\sum_{k\in U^B_i}\pi^{AB}_{(i,k)}=1, \forall i.$$

Since, for all \( i \), the expression  
\[
\sum_{k\in U^B_i} \frac{(\theta^{AB}_{ik})^2 \sum_{q=1}^Q \alpha^B_q (\mathbf{x}_{q,k}^B)^2}{\pi^A_i \pi^{AB}_{(i,k)}}
\]  
depends solely and specifically on the set of parameters constrained by the condition \( \sum_{k \in U^B_i} \pi^{AB}_{(i,k)} = 1 \), the overall optimization problem can be decomposed into \( N^A \) independent subproblems. By omitting the constant factor \( \pi^A_i \), which plays no role in the optimization, and applying standard Lagrangian techniques, we directly obtain Proposition~\ref{PiOpt}.

\section{Proof of Proposition \ref{Probacible}}\label{ProofProbacible}

We make use of the properties outlined in Section~\ref{DSD}, and introduce two additional lemmas that will be instrumental in establishing the final proof. In this section, we consider two equivalent representations of a random sample $\mathbb{S}$. Consistent with the general framework of this article, $\mathbb{S}$ may either denote the list of integers indexing the selected units, or a vector in $[0,1]^N$, where the $i$-th entry equals $1$ if the corresponding unit is included in the sample and $0$ otherwise. As a result, we may allow for slight abuses of notation, which should nevertheless cause no ambiguity.
\subsection{Two useful lemmas}
\begin{lemma}
Under assumptions \ref{H3-4}, $\mathcal{P}^{AB}$ is determinantal.
\end{lemma}

\begin{proof}
Under Assumptions~\ref{H3-4},  \( \mathcal{P}^{AB} \) is fully characterized by the first-order inclusion probabilities \( \mathbf{\Pi}^{1AB} \), since for any subset \( s^{AB} \in 2^{U^A \times U^B} \), we have:
\[
\Pr(s^{AB} \subseteq \mathbb{S}^{AB}) = 
\begin{cases}
0 & \text{if } \exists\, ((i,k), (j,l)) \in s^{AB} \times s^{AB} \text{ such that } i = j \text{ and } k \neq l, \\
\underset{\{(i,k) \in s^{AB}\}}{\prod} \pi^{AB}_{(i,k)} & \text{otherwise}.
\end{cases}
\]Moreover, up to a permutation, \( \mathcal{P}^{AB} \) is a one-per-stratum design with \( H = N_A \) strata of size \( N_B \). Each stratum \( h \in \{1, \ldots, H\} \) is defined by the set \( \{(h,k): k \in \{1, \ldots, N_B\} \} \). Up to this permutation, define \( \mathbf{K}^{AB} \) as a block-diagonal matrix of size \( N^A N^B \times N^A N^B \), consisting of \( H = N^A \) blocks, each of dimension \( N^B \times N^B \). The \( h \)-th block, denoted \( \mathbf{K}^{h,AB} \), is defined by:
\[
\mathbf{K}^{h,AB}_{kl} =
\begin{cases}
\pi^{AB}_{(h,k)} & \text{if } k = l, \\
\sqrt{\pi^{AB}_{(h,k)} \, \pi^{AB}_{(h,l)}} & \text{if } k \neq l.
\end{cases}
\]

Under Assumptions~\ref{H3-4}, each matrix \( \mathbf{K}^{h,AB} \) is a projection matrix, and so is the full matrix \( \mathbf{K}^{AB} \). Therefore, there exists a determinantal sampling design \( DSD(\mathbf{K}^{AB}) \) whose first-order and second-order inclusion probabilities match those specified in Assumptions~\ref{H3-4}. As a consequence, we conclude that:
\[
DSD(\mathbf{K}^{AB}) = \mathcal{P}^{AB}.
\]
\end{proof}
\begin{lemma}
Let \( DSD(\mathbf{K}^1) \) and \( DSD(\mathbf{K}^2) \) be two independent determinantal sampling designs over a population of size \( N \), with \( DSD(\mathbf{K}^2) \) being a Poisson sampling. Let \( \mathbb{S}_1 \sim DSD(\mathbf{K}^1) \) and \( \mathbb{S}_2 \sim DSD(\mathbf{K}^2) \). Consider the vector representation of \( \mathbb{S}_1 \) (resp. \( \mathbb{S}_2 \)), where each entry \( k \in \{1, \ldots, N\} \) takes the value \( 1 \) if \( k \in \mathbb{S}_1 \) (resp. \( \mathbb{S}_2 \)), and \( 0 \) otherwise. Then, the Hadamard (entrywise) product \( \mathbb{S}_1 \odot \mathbb{S}_2 \) follows a determinantal sampling design with kernel \( \mathbf{K}^1 \mathbf{K}^2 \), that is:
\[
\mathbb{S}_1 \odot \mathbb{S}_2 \sim DSD(\mathbf{K}^1 \mathbf{K}^2).
\] 
\end{lemma}
\begin{proof}
$\mathbf{K}^1$ and $\mathbf{K}^2$ are contracting matrices, so is $\mathbf{K}^1\mathbf{K}^2$. Let $s\in 2^U$, we have: $\Pr(s \subseteq  \mathbb{S}_1 \odot \mathbb{S}_2)=\Pr(s \subseteq \mathbb{S}_1,s\subseteq  \mathbb{S}_2)=\Pr(s \subseteq \mathbb{S}_1 )\Pr(s \subseteq  \mathbb{S}_2)=\det(\mathbf{K}^1_{|s})\det(\mathbf{K}^2_{|s})=\det(\mathbf{K}^1_{|s}\mathbf{K}^2_{|s})=det((\mathbf{K}^1\mathbf{K}^2)_{|s})$, since $\mathbf{K}^2$ is diagonal.

\end{proof}

\subsection{Proof of Proposition \ref{Probacible}}

Let $\mathbb{S}^A_{|U^A_k}$ be the restriction of $\mathbb{S}^A$ to $U^A_k$. $\mathbb{S}^A_{|U^A_k}$ is a determinantal random with kernel $\mathbf{K}^A_{|U^A_k}$. 
Let $\mathbb{S}^{AB}_{|(U^{A}_k,k)}$ be the restriction of $\mathbb{S}^{AB}$ to $(U^{A}_k,k)$. $\mathbb{S}^{AB}_{|(U^{A}_k,k)}$ is a determinantal random sample with kernel $\mathbf{K}^{AB}_{|(U^{A}_k,k)}$, that is diagonal. Both sampling design are formally defined over the same population $U^A_k$. 
\paragraph{}
Unit $k$ will be selected in $\mathbb{S}^{2B}$ if $\mathbb{S}^A_{|U^A_k} \odot \mathbb{S}^{AB}_{|(U^{A}_k,k)}$ is different from the empty set. That is equivalent to $U^A_k$ is not fully selected in the complementary of $\mathbb{S}^A_{|U^A_k}\odot\mathbb{S}^{AB}_{|(U^{A}_k,k)}$, which is determinantal with kernel $\mathbf{I}_{N^A_k}-\mathbf{K}^A_{|U^A_k}\mathbf{K}^{AB}_{(|U^{A}_k,k)}$, as a consequence:
\paragraph{}
$$\pi^{2B}_k=\Pr(\{k\} \subseteq \SS^{2B})=1-\Pr(U^A_k \subseteq \left ( \mathbb{S}^A_{|U^A_k}\odot\mathbb{S}^{AB}_{|(U^{A}_k,k)}\right )^c)=1-\det(\mathbf{I}_{N^A_k}-\mathbf{K}^A_{|U^A_k}\mathbf{K}^{AB}_{(|U^{A}_k,k)}).$$

\paragraph{}

For the joint probabilities, we write $\pi^{2B}_{kl}=\Pr(\{k,l\} \subseteq \SS^{2B})=\Pr(\{k\} \subseteq \SS^{2B})+\Pr(\{l\} \subseteq \SS^{2B})-\Pr(\{k\} \subseteq \SS^{2B}\text{ or }\{l\} \subseteq \SS^{2B})$. The first two terms are given by the previous result. We therefore focus on $\Pr(\{k\} \subseteq \SS^{2B}\text{ or }\{l\} \subseteq \SS^{2B})$:
$$
\begin{array}{lll}
&=&\Pr(\{\exists i \in U^A_k \mid \{i,(i,k)\}\in \SS^A\times\SS^{AB}\}\text{ or }\{\exists j \in U^A_l \mid \{j,(j,l)\}\in \SS^A\times\SS^{AB}\})\\

&=&\Pr(\{\exists i \in U^A_k \cup  U^A_l \mid \{i,(i,k)\}\in \SS^A\times\SS^{AB}\}\text{ or }\{\exists j \in U^A_k \cup U^A_l \mid \{j,(j,l)\}\in \SS^A\times\SS^{AB}\})\\
&=&\Pr(\{\exists i \in U^A_k \cup  U^A_l \mid \{i \in \mathbb{S}^A\}\text{ and }\{\exists i \in U^A_k \cup  U^A_l \mid (i,k)\text{ or }(i,l) \in \SS^{AB}\})\\
\end{array}$$

Under Assumptions~\ref{H3-4}, both sub-samples \( \mathbb{S}^{AB}_{|(U^A_k \cup U^A_l, k)} \) and \( \mathbb{S}^{AB}_{|(U^A_k \cup U^A_l, l)} \) are Poisson random samples. In this specific case, their sum
\[
\mathbb{S}^{AB}_{|(U^A_k \cup U^A_l, k)} + \mathbb{S}^{AB}_{|(U^A_k \cup U^A_l, l)}
\]
is also a Poisson random sample, with kernel 
\[
\mathbf{K}^{AB}_{|(U^A_k \cup U^A_l, k)} + \mathbf{K}^{AB}_{|(U^A_k \cup U^A_l, l)}.
\]

As a consequence, the event 
\[
\left\{ \exists i \in U^A_k \cup U^A_l \mid i \in \mathbb{S}^A \right\} \quad \text{and} \quad \left\{ \exists i \in U^A_k \cup U^A_l \mid (i,k) \in \mathbb{S}^{AB} \text{ or } (i,l) \in \mathbb{S}^{AB} \right\}
\]
is equivalent to the event that the vector
\[
\mathbb{S}^{A}_{U^A_k \cup U^A_l} \odot \left( \mathbb{S}^{AB}_{|(U^A_k \cup U^A_l, k)} + \mathbb{S}^{AB}_{|(U^A_k \cup U^A_l, l)} \right)
\]
has at least one non-zero entry. This situation mirrors the one encountered when evaluating \( \pi^{2B}_k \), leading to:
\[
\Pr\left( \{k\} \subseteq \mathbb{S}^{2B} \text{ or } \{l\} \subseteq \mathbb{S}^{2B} \right) = 1 - \det \left( \mathbf{I}_{N^A_{k\ell}} - \mathbf{K}^A_{|U^A_k \cup U^A_l} \left( \mathbf{K}^{AB}_{|(U^A_k \cup U^A_l, k)} + \mathbf{K}^{AB}_{|(U^A_k \cup U^A_l, \ell)} \right) \right),
\]
and ultimately to the final result.

\section{Greedy Optimization Algorithm over Orthogonal Projection Matrices with Prescribed Diagonal}\label{Greedy}
\begin{theorem}[\citet{loonis2019determinantal}]\label{TheoLO}
Let $\mathbf{K}$ be a Hermitian contracting matrix, $(i,j)\in U\times U$ such that $K_{ii} \neq K_{jj}$ for $i\neq j$ and $K_{ij} \neq 0$. Let $W_{ij}(\theta)$ be the Givens unitary
operator whose matrix relative to the canonical basis has $\cos\theta$ at the $(i,i)$ and $(l, l)$ entries, $-\sin\theta$ and $\sin\theta$ at the $(i,j)$ and $(j,i)$ entries, respectively, $1$ at all other diagonal entries, and $0$ at all other off-diagonal entries, where:

$$t=\frac{2Re(K_{ij})}{K_{ii}-K_{jj}}, \cos \theta = \frac{1}{\sqrt{1+t^2}} \text{ and }\sin \theta = t \cos \theta.$$

Then matrix $\mathbf{K}(\theta)=W_{ij}(\theta)\mathbf{K}W_{ij}(\theta)^{\intercal}$ has the same diagonal and spectrum as $\mathbf{K}$.
\end{theorem}

\begin{algorithm}[H]
\caption{Iterative optimization over the set of Hermitian matrices with prescribed spectrum and diagonal via Givens rotations \citep{loonis2019determinantal}}
\label{algo:rotation-optimization}
\begin{algorithmic}[1]
\Require Auxiliary information matrix \( X \), initial contracting matrix \( \mathbf{K}_0 \), objective function \( C_X(\mathbf{K}) \), number of iterations \( R \)
\Ensure Optimized contracting matrix \( \mathbf{K}^* \)
\State Set \( \mathbf{K} \gets \mathbf{K}_0 \)
\For{$r = 1$ to $R$}
    \For{$i = 1$ to $N-1$}
        \For{$j = i+1$ to $N$}
            \If{ \( K_{ii} \neq K_{jj} \) \textbf{and} \( K_{ij} \neq 0 \) }
                \State Compute the Givens rotation angle \( \theta^r_{ij} \) according to Theorem~\ref{TheoLO}
                \State Set \( \mathbf{K}' \gets W_{ij}(\theta^r_{ij})\, \mathbf{K}\, W_{ij}(\theta^r_{ij})^\intercal \)
                \If{ \( C_X(\mathbf{K}') < C_X(\mathbf{K}) \) }
                    \State Update \( \mathbf{K} \gets \mathbf{K}' \)
                \EndIf
            \EndIf
        \EndFor
    \EndFor
\EndFor
\State \Return \( \mathbf{K}^* \gets \mathbf{K} \)
\end{algorithmic}
\end{algorithm}
If $\mathbf{K}_0$ is a projection matrix, then $\mathbf{K}^*$ is a projection matrix with the same diagonal as $\mathbf{K}_0$.

\section{Figures}\label{Figures}

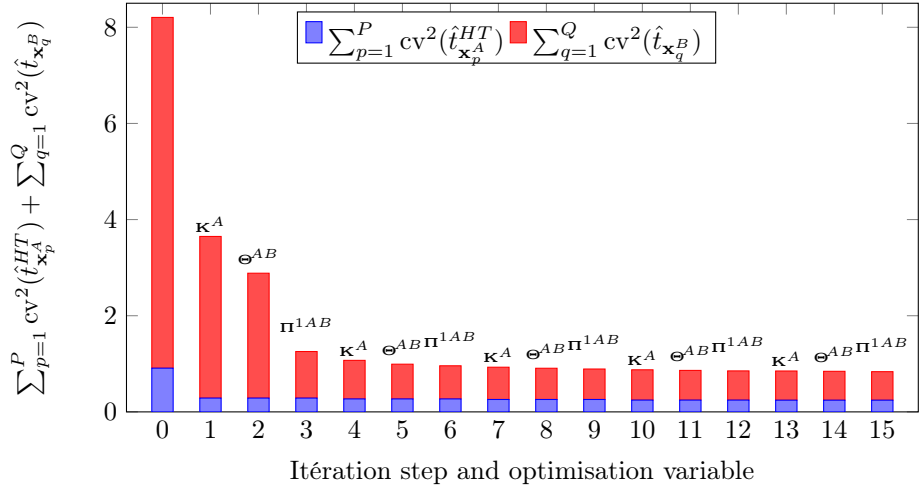
\begin{figure}[h!]
\begin{tikzpicture}
  \begin{axis}[
    ybar stacked,
    bar width=8pt,
    ymin=0,
    ymax=8.5, % ajusté à la plus grande somme empilée
    xtick=data,
    xticklabels={0,...,15},
    xlabel={Itération step and optimisation variable},
    ylabel={$\sum_{p=1}^P \operatorname{cv}^2(\hat{t}^{HT}_{\mathbf{x}^A_p})+ \sum_{q=1}^Q \operatorname{cv}^2(\hat{t}_{\mathbf{x}^B_q})$},
    width=0.8\textwidth,
    height=7cm,
    enlarge x limits=0.05,
    legend style={at={(0.48,0.98)}, anchor=north, legend columns=2}
  ]
  
  % Série 1 (en bas)
  \addplot+[fill=blue!50] coordinates {
    (0,0.9091326) (1,0.2868835) (2,0.2868835) (3,0.2868835)
    (4,0.2694231) (5,0.2694231) (6,0.2694231) (7,0.2580023)
    (8,0.2580023) (9,0.2580023) (10,0.2450473) (11,0.2450473)
    (12,0.2450473) (13,0.2432377) (14,0.2432377) (15,0.2432377)
  };

  % Série 2 (au-dessus)
  \addplot+[fill=red!70] coordinates {
    (0,7.2957885) (1,3.3624090) (2,2.5977644) (3,0.9688115)
    (4,0.8023343) (5,0.7228812) (6,0.6895573) (7,0.6714423)
    (8,0.6497847) (9,0.6338751) (10,0.6304729) (11,0.6178024)
    (12,0.6075587) (13,0.6074171) (14,0.6008512) (15,0.5932753)
  };

  \legend{$\sum_{p=1}^P \operatorname{cv}^2(\hat{t}^{HT}_{\mathbf{x}^A_p})$, $\sum_{q=1}^Q \operatorname{cv}^2(\hat{t}_{\mathbf{x}^B_q})$}
\node at (axis cs:1, 0.2868835 + 3.3624090 + 0.25) {\tiny $\mathbf{K}^A$};
  \node at (axis cs:4, 0.2694231 + 0.8023343 + 0.25) {\tiny $\mathbf{K}^A$};
  \node at (axis cs:7, 0.2580023 + 0.6714423 + 0.25) {\tiny $\mathbf{K}^A$};
  \node at (axis cs:10, 0.2450473 + 0.6304729 + 0.25) {\tiny $\mathbf{K}^A$};
  \node at (axis cs:13, 0.2432377 + 0.6074171 + 0.25) {\tiny $\mathbf{K}^A$};

  \node at (axis cs:2, 0.2868835 + 2.5977644 + 0.35) {\tiny $\mathbf{\Theta}^{\!AB}$};
  \node at (axis cs:5, 0.2694231 + 0.7228812 + 0.35) {\tiny $\mathbf{\Theta}^{\!AB}$};
  \node at (axis cs:8, 0.2580023 + 0.6497847 + 0.35) {\tiny $\mathbf{\Theta}^{\!AB}$};
  \node at (axis cs:11, 0.2450473 + 0.6178024 + 0.35) {\tiny $\mathbf{\Theta}^{\!AB}$};
  \node at (axis cs:14, 0.2432377 + 0.6008512 + 0.35) {\tiny $\mathbf{\Theta}^{\!AB}$};

  \node at (axis cs:3, 0.2868835 + 0.9688115 + 0.55) {\tiny $\mathbf{\Pi}^{1AB}$};
  \node at (axis cs:6, 0.2694231 + 0.6895573 + 0.55) {\tiny $\mathbf{\Pi}^{1AB}$};
  \node at (axis cs:9, 0.2580023 + 0.6338751 + 0.55) {\tiny $\mathbf{\Pi}^{1AB}$};
  \node at (axis cs:12, 0.2450473 + 0.6075587 + 0.55) {\tiny $\mathbf{\Pi}^{1AB}$};
  \node at (axis cs:15, 0.2432377 + 0.5932753 + 0.55) {\tiny $\mathbf{\Pi}^{1AB}$};
  \end{axis}
\end{tikzpicture}
\caption{Values of $\sum_{p=1}^P \operatorname{cv}^2(\hat{t}^{HT}_{\mathbf{x}^A_p})+ \sum_{q=1}^Q \operatorname{cv}^2(\hat{t}_{\mathbf{x}^B_q})$ after each optimisation step.}\label{Fig_opt}
\end{figure}

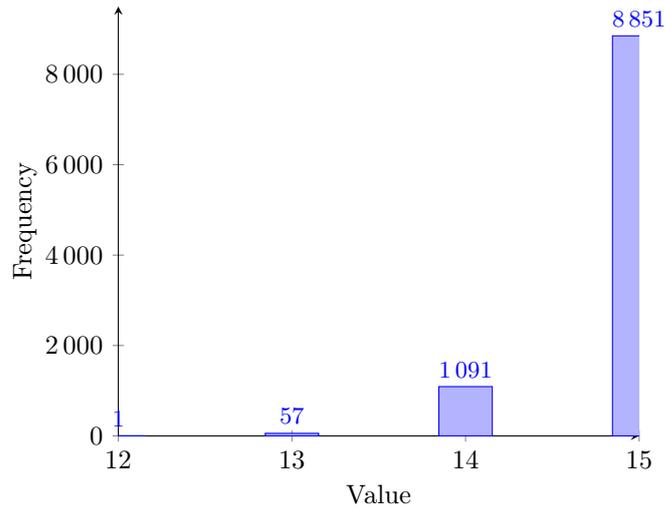
\begin{figure}[h!]
\centering

\begin{tikzpicture}
\begin{axis}[
    ybar,
    bar width=20pt,
    xlabel={Value},
    ylabel={Frequency},
    xtick=data,
    symbolic x coords={12, 13, 14, 15},
    ymin=0,
    ymax=9500,
    nodes near coords,
    enlarge x limits=0.2,
    axis lines=left,
    scaled y ticks=false,
    % Format des nombres partout
    /pgf/number format/use comma=false,
    /pgf/number format/1000 sep={\,},
    every node near coord/.append style={
        font=\small,
        /pgf/number format/use comma=false,
        /pgf/number format/1000 sep={\,}
    }
]
\addplot coordinates {(12,1) (13,57) (14,1091) (15,8851)};
\end{axis}
\end{tikzpicture}
\caption{Breakdown of target sample sizes across 10 000 $\mathbb{S}^{2B}$ simulations with a fixed intermediate sample size of 15.}\label{samplesize}
\end{figure}
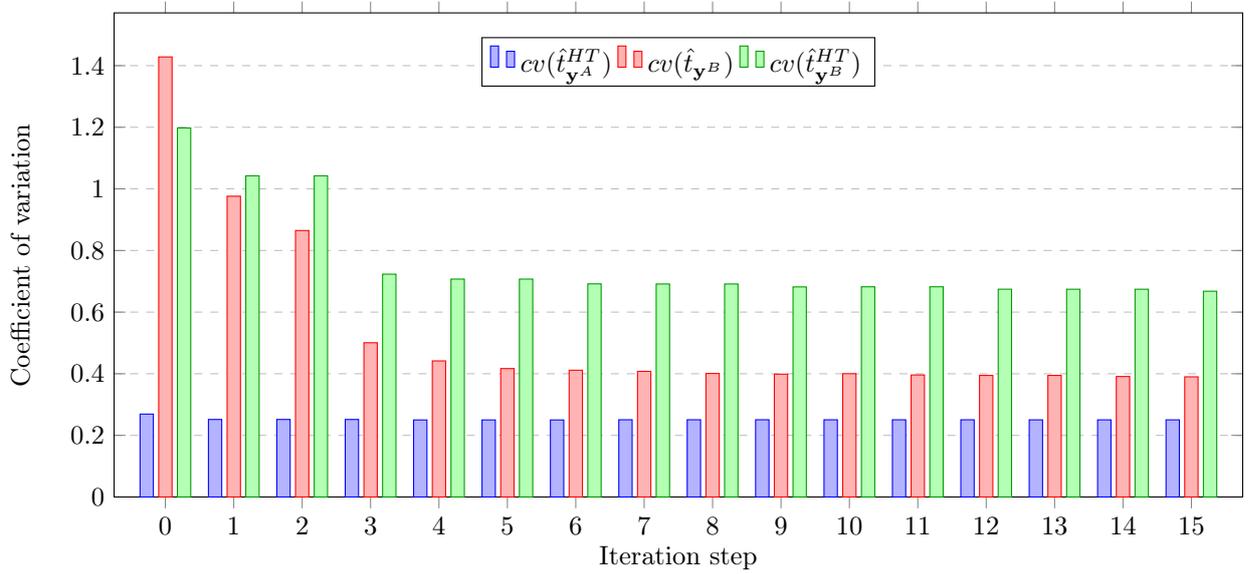
\begin{figure}[h!]
\centering
\begin{tikzpicture}
\begin{axis}[
    ybar,
    bar width=5pt,
    width=\textwidth,
    height=8cm,
    x=0.9cm, % <- augmente l'espacement entre les groupes
    enlarge x limits={0.02, 0.05}, % <- réduit espace avant le 1er groupe
    ylabel={Coefficient of variation},
    xlabel={Iteration step},
    symbolic x coords={0,1,2,3,4,5,6,7,8,9,10,11,12,13,14,15},
    xtick=data,
    legend style={at={(0.5,0.95)},
      anchor=north,legend columns=3},
    ymin=0,
    ymajorgrids=true,
    grid style=dashed
]

\addplot+[blue, fill=blue!30] coordinates {
  (0,0.2691) (1,0.2516) (2,0.2516) (3,0.2516)
  (4,0.2501) (5,0.2501) (6,0.2501) (7,0.2507)
  (8,0.2507) (9,0.2507) (10,0.2506) (11,0.2506)
  (12,0.2506) (13,0.2504) (14,0.2504) (15,0.2504)
};

\addplot+[red, fill=red!30] coordinates {
  (0,1.4283) (1,0.9762) (2,0.8649) (3,0.5006)
  (4,0.4417) (5,0.4168) (6,0.4115) (7,0.4078)
  (8,0.4009) (9,0.3987) (10,0.4002) (11,0.3960)
  (12,0.3945) (13,0.3946) (14,0.3910) (15,0.3900)
};

\addplot+[green!60!black, fill=green!30] coordinates {
  (0,1.1978) (1,1.0425) (2,1.0425) (3,0.7232)
  (4,0.7072) (5,0.7072) (6,0.6917) (7,0.6913)
  (8,0.6913) (9,0.6818) (10,0.6822) (11,0.6822)
  (12,0.6744) (13,0.6743) (14,0.6743) (15,0.6678)
};

\legend{$cv(\hat{t}^{HT}_{\mathbf{y}^A})$, $cv(\hat{t}_{\mathbf{y}^B})$, $cv(\hat{t}^{HT}_{\mathbf{y}^B})$}

\end{axis}
\end{tikzpicture}
\caption{Estimation precision of $t_{\mathbf{y}^A}$ and $t_{\mathbf{y}^B}$ over iterations using Horvitz-Thompson or GWSM.}\label{Fig_opt_interet}
\end{figure}

\end{document}